\documentclass[12pt]{article}
\usepackage[cp1250]{inputenc}
\usepackage[verbose,a4paper,tmargin=0.5in,lmargin=1in,rmargin=1in]{geometry}
\usepackage[pdftex]{graphicx}
\usepackage{amsfonts}
\usepackage{indentfirst}
\usepackage{latexsym}
\usepackage{amsmath}
\usepackage{amsthm}
\usepackage{amssymb}
\usepackage{psfrag}
\usepackage{eufrak}
\usepackage[mathscr]{eucal} 
\usepackage{color}
\usepackage{cite}
\usepackage{verbatim}
\usepackage{graphicx}
\usepackage[all]{xy}
\usepackage{float}

\newtheorem{XX}{Twierdzenie}[section]
\newtheorem{Twierdzenie}[XX]{Theorem}
\newtheorem{Definicja}[XX]{Definition}

\newtheorem{Wniosek}[XX]{Corollary}

\newfloat{Scheme}{tbp}{muj}

\title{On geometry of congruences of null strings in 4-dimensional complex and real pseudo-Riemannian spaces.}

\author{$\textrm{Adam Chudecki}^{*}$}

\begin{document}

\maketitle

$*$ Center of Mathematics and Physics, Lodz University of Technology, 
\newline
$\ \ \ \ \ $ Al. Politechniki 11, 90-924 Łódź, Poland, adam.chudecki@p.lodz.pl
\newline
\newline
\newline
\textbf{Abstract}. 
\newline
4-dimensional spaces equipped with 2-dimensional (complex holomorphic or real smooth) completely integrable distributions are considered. The integral manifolds of such distributions are totally null and totally geodesics 2-dimensional surfaces which are called \textsl{the null strings}. Properties of congruences (foliations) of such 2-surfaces are studied. Some relations between properties of congruences of null strings, Petrov-Penrose type of SD Weyl spinor and algebraic types of the traceless Ricci tensor are analyzed. 

%#####################################################################################

\section{Introduction}

The present paper is devoted to some geometrical aspects of complex and real 4-dimensional spaces. We assume that  spaces are equipped with the holomorphic (in complex case) or smooth (in real case) metric. Except the metric tensor we assume the existence of 2-dimensional, completely integrable distributions. Their integral manifolds are totally null and totally geodesic 2-dimensional surfaces, \textsl{the null strings}. The family of such surfaces constitute \textsl{the congruence (foliation) of the null strings}. 2-dimensional distribution is related to the corresponding 2-form, which is  self-dual (SD) or anti-self-dual (ASD). In this paper we analyze the SD congruences of the null strings. In what follows we use abbreviation \textsl{cns} (cns := congruence of the SD null strings).

The idea of $m$-dimensional real spaces equipped with $n$-dimensional ($n < m$) totally null, parallely propagated distributions appeared in the fifties \cite{Walker} and such spaces are known nowadays as \textsl{Walker spaces}. For our purposes the most interesting is the case with $m=4$ and $n=2$, but in general we do not assume that the distribution is parallel. In dimension 4, the only real spaces which admit cns are spaces equipped with the metric of the neutral (split, ultrahyperbolic) signature $(++--)$. Moreover, such structures are admitted by 4-dimensional complex manifolds. 

In complex geometry cns play a great role. These objects appear in the theory of heavenly spaces ($\mathcal{H}$-spaces) and in the twistor theory as a $\alpha$($\beta$)-surfaces. The relation between cns and the algebraic degeneration of SD Weyl spinor in Einstein spaces has been found in \cite{Plebanski_surf}. The further analysis of the spaces which admit cns led to the concept of the hyperheavenly spaces ($\mathcal{HH}$-spaces). Some properties of cns have been studied in \cite{Boyer_Finley_Plebanski,Przanowski} but the significant progress in understanding the geometry of cns has been done in \cite{Rozga, Rozga_Robinson}. In \cite{Rozga} the properties of cns are related to the properties of shear-free null geodesic congruences (abbreviate by \textsl{sngc}). Clearly, cns are complex analogs of sngc. They have been considered as the most important structures in so-called \textsl{Plebański programme}. The main idea of the Plebański programme was to find the general techniques of generating the real solutions of the Einstein field equations from the complex ones. Unfortunately, such techniques are still unknown despite over 40 years of thorough investigations. 

Real 4-dimensional Walker spaces equipped with 2-dimensional integrable and totally null distributions have been recently analyzed intensively (see \cite{Chaichi,Ch_Walkery} and references therein). It appeared, that the most natural formalism in description of the 4-dimensional Walker spaces is the spinorial formalism. Spinorial formalism as a mathematical tool has been developed since the sixties \cite{Pirani, Penrose, Plebanski_Spinors,Pleban_formalism_1}. Spaces with both SD and ASD 2-dimensional integrable distributions have been introduced in \cite{Ch_Walkery} (see also \cite{Law_Matsushita}). Recently some 4-dimensional spaces equipped with the metric of neutral signature and two distinct cns appeared naturally in the works devoted to geometrical model of bodies which roll on each other without slipping and twisting \cite{Nurowski_An,Nurowski_Bor_Lamoneda}. 

Cns play an important role in Lorentzian geometry as well. The main idea of distinguished papers \cite{Trautman, Nurowski_Trautman} was to treat sngc as an intersection of the SD and ASD congruences of the null strings which exist in the complexification of the tangent bundle. In \cite{Trautman} \textsl{Robinson manifold} have been introduced and their relation with cns has been pointed out. Also the SD null string equations (\ref{rownanie_struny_SD}) have been used to obtain a special class of solutions of the Rarita-Schwinger equation \cite{Szereszewski_Tafel}.

Cns play also a great role in the generalizations of the Goldberg - Sachs theorem \cite{Goldberg_Sachs}. 
For complex, 4-dimensional Einstein spaces the analog of Goldberg - Sachs theorem was presented in \cite{Plebanski_surf}. Then this idea has been generalized to the case of nonzero traceless Ricci tensor \cite{Przanowski_generalized_Goldberg} and to the lower \cite{Nurowski_Chabert} and higher dimensions \cite{Taghavi_Chabert_2}. In \cite{Rod_Hill_Nurowski} the authors considered "sharp" versions of the Goldberg - Sachs theorem in dimension 4.

All mentioned above facts prove that cns carry a very important structure which deserves closer analysis. Our previous work \cite{Chudecki_para_Hermite_para_Kahler} was devoted to Einstein para-Hermite and Einstein para-K\"{a}hler spaces. Such spaces are equipped with two distinct cns (see also \cite{Aleeksiewicz} and references therein).
In this paper we analyze the non-Einsteinian case. There are two main goals of the present paper. The first is to analyze the relation between objects which characterize cns (expansion, Sommers vector) and algebraic types of the traceless Ricci tensor (in complex case small portion of such analysis has been presented in \cite{Przanowski}). The second aim is to analyze the spaces equipped with three \cite{Rozga_Robinson} or four distinct cns. 

Our considerations are purely local. The results are valid for the 4-dimensional complex manifolds and can be easily carried over to the case of the real manifold with neutral signature metric. Usually it is enough to replace all the holomorphic functions by real smooth functions and complex coordinates by the real ones. Sometimes there appear subtle differences between complex and real spaces and these differences are pointed out and thoroughly analyzed.

The paper is organized, as follows.

Section \ref{sekcja_formalizm} contains basic information on the spinorial formalism. We use spinorial formalism in Infeld - Van der Waerden - Plebański notation rather then Penrose notation. In section \ref{Spaces_equipped_with_one_congruence_of_the_SD_null_strings} we consider the spaces which admit one cns. Relation between properties of the cns and algebraic types of the traceless Ricci tensor are pointed out. Finally a new Theorem \ref{Twierdzenie_nasze_o_zerowych_wektorach_wlasnych} concerning the relation between cns and null eigenvectors of the traceless Ricci tensor is proved. Particularly, we find all algebraic types of the traceless Ricci tensor generated by the one nonexpanding cns. Subsection devoted to the Weyl spinor is less original and it contains results known earlier. However, these results are spread all over the wide literature, so we present them here for completeness.

In section \ref{Sekcja_dwie_komplementarne_koonggruencje_strun} we analyze the spaces equipped with two distinct cns. It appears, that two distinct cns determine the traceless Ricci tensor completely (\ref{definition_of_spinors_in_tracelessRicci}) and it is the main result of this section. All algebraic types of the traceless Ricci tensor generated by two distinct cns are found. Finally, the relation between null eigenvectors of the traceless Ricci tensor and properties of two distinct cns is proved. 

Section \ref{sekcja_o_dodatkowych_kongruencjach} is devoted to spaces which admit a richer structure then the para-Hermite spaces, i.e. they admit the existence of three or four distinct cns.  We find the general form of the metric which admits two expanding cns and one nonexpanding cns (metric (\ref{metryka_z_trzema_kongruencjami_SDstrun_jedna_nonexpanding})). Finally, we consider the case which - according to our best knowledge - has not been considered earlier: the space equipped with four distinct cns. In this case we arrive at the metric (\ref{metryka_przestrzeni_z_czterema_kongruencjami_strun}) with constraint equations (\ref{rownania_jakie_zostaly_dla_metryki_cztery_kongr}). Concluding remarks end the paper.

%#####################################################################################

\setlength\arraycolsep{2pt}
\setcounter{equation}{0}

\section{Formalism}
\label{sekcja_formalizm}

In this section we present foundations of the spinorial formalism which seem to be crucial in geometry based on cns. For more detailed treatment see \cite{Pleban_formalism_1, Pleban_formalism_2, Plebanski_Spinors}.

Let $\mathcal{M}$ be a 4-dimensional complex analytic differentiable manifold endowed with a holomorphic metric $ds^2$ or a real 4-dimensional smooth differentiable manifold endowed with a real smooth metric $ds^2$. Thus one deals with \textsl{complex relativity} ($\mathbf{CR}$) or with \textsl{real relativity} ($\mathbf{RR}$). More precisely, inside real relativity we distinguish \textsl{hyperbolic relativity} ($\mathbf{HR}_{+}$ if signature of the metric is $(+++-)$ or $\mathbf{HR}_{-}$ if the signature is $(+---)$), \textsl{ultrahyperbolic relativity} $\mathbf{UR}$ (\textsl{neutral, split}, signature $(++--)$) and \textsl{Euclidean relativity} ($\mathbf{ER}_{+}$ or $\mathbf{ER}_{-}$ if signature is $(++++)$ or $(----)$, respectively).

Denote by $\Lambda^{p}(\mathcal{M})$ the space of all p-forms on $\mathcal{M}$ ($p=0,1,2,3,4$). Let $(e^1, e^2, e^3, e^4) \in \Lambda^{1} (\mathcal{M})$ be the members of the null tetrad so the metric can be written in the form
\begin{equation}
ds^{2} = 2e^{1}e^{2} + 2e^{3}e^{4}
\end{equation}
The bases of the $\Lambda^{p}(\mathcal{M})$ are given by ($A=1,2$,$ \dot{B}= \dot{1}, \dot{2}$)
\begin{eqnarray}
\label{formy_spinorowo}
\Lambda^{0}(\mathcal{M}):  && 1
\\ \nonumber
\Lambda^{1}(\mathcal{M}): && g^{A\dot{B}}
\\ \nonumber
\Lambda^{2}(\mathcal{M}): && S^{AB} := \frac{1}{2} \in_{\dot{R}\dot{S}} g^{A\dot{R}} \wedge g^{B\dot{S}} \ , \ \ \ S^{\dot{A}\dot{B}} := \frac{1}{2} \in_{RS} g^{R\dot{A}}\wedge g^{S\dot{B}}
\\ \nonumber
\Lambda^{3}(\mathcal{M}): && \check{g}^{A\dot{B}} := \frac{1}{3} g^{A}_{\ \; \dot{C}} \wedge S^{\dot{C}\dot{B}} \equiv \frac{1}{3} S^{A}_{\ \; C} \wedge g^{C \dot{B}}
\\ \nonumber
\Lambda^{4}(\mathcal{M}): && vol := -\frac{1}{12} S^{AB} \wedge S_{AB} \equiv \frac{1}{12} S^{\dot{A}\dot{B}} \wedge S_{\dot{A}\dot{B}} \equiv \frac{1}{8} \check{g}^{A\dot{B}} \wedge g_{A\dot{B}}
\end{eqnarray}
where
\begin{subequations}
\begin{eqnarray}
\label{1_formy_spinorowe}
&&(g^{A\dot{B}}) := \sqrt{2}
\left[\begin{array}{cc}
e^4 & e^2 \\
e^1 & -e^3
\end{array}\right] 
\\ 
&& (S^{AB}) := 
\left[\begin{array}{cc}
2 e^4 \wedge e^2 & e^1 \wedge e^2 + e^3 \wedge e^4 \\
e^1 \wedge e^2 + e^3 \wedge e^4 & 2 e^3 \wedge e^1
\end{array}\right] 
\\ 
&&(S^{\dot{A} \dot{B}}) := 
\left[\begin{array}{cc}
2 e^4 \wedge e^1 & -e^1 \wedge e^2 + e^3 \wedge e^4 \\
-e^1 \wedge e^2 + e^3 \wedge e^4 & 2 e^3 \wedge e^2
\end{array}\right] 
\\ 
&& (\check{g}^{A\dot{B}}) := \sqrt{2}
\left[\begin{array}{cc}
e^1 \wedge e^2 \wedge e^4 \ \ \ & \ \ \ e^2 \wedge e^3 \wedge e^4 \\
-e^1 \wedge e^3 \wedge e^4 \ \ \ & \ \ \ e^1 \wedge e^2 \wedge e^3
\end{array}\right] 
\\ 
&& vol := e^1 \wedge e^2 \wedge e^3 \wedge e^4
\end{eqnarray}
\end{subequations}
The spinorial indices are manipulated according to the rules
\begin{eqnarray}
\label{reguly_podnoszenia_wskaznikow}
\Psi_{A} = \in_{AB} \Psi^{B} \ \ \ , \ \ \ \Psi_{\dot{A}} = \in_{\dot{A}\dot{B}} \Psi^{\dot{B}}
\ \ \ , \ \ \
\Psi^{A} = \in^{BA} \Psi_{B} \ \ \ , \ \ \ \Psi^{\dot{A}} = \in^{\dot{B}\dot{A}} \Psi_{\dot{B}}
\\ \nonumber
( \in_{AB} ) := \left[ \begin{array}{cc}
                            0 & 1   \\
                           -1 & 0  
                            \end{array} \right] =: ( \in^{AB} )
\ \ \ , \ \ \ 
( \in_{\dot{A}\dot{B}} ) := \left[ \begin{array}{cc}
                            0 & 1   \\
                           -1 & 0  
                            \end{array} \right] =: ( \in^{\dot{A}\dot{B}} ) \ \ \ \ 
\end{eqnarray}
Nevertheless, one has to carefully raise and lower spinorial indices in the objects from tangent space. Indeed, if $\partial^{\dot{A}}$ means $\partial / \partial \Psi_{\dot{A}}$ and $\Psi_{\dot{A}} = \in_{\dot{A}\dot{B}} \Psi^{\dot{B}}$ then consistency with (\ref{reguly_podnoszenia_wskaznikow}) implies
\begin{equation}
\label{podnoszenie_wskaznikow_spinorowych}
\partial^{\dot{A}} = \partial_{\dot{B}} \in^{\dot{A} \dot{B}} \ , \ \ \ \partial_{\dot{A}} = \in_{\dot{B} \dot{A}} \partial^{\dot{B}} \ , \ \ \ \partial^{A} = \partial_{B} \in^{A B} \ , \ \ \ \partial_{A} = \in_{B A} \partial^{B}
\end{equation}
The 1-forms $g^{A \dot{B}}$ are holomorphic in $\mathbf{CR}$. In the case of $\mathbf{RR}$ the 1-forms $g^{A \dot{B}}$ have the following properties under complex conjugation:
\begin{eqnarray}
\mathbf{HR}_{+}: \ \ && \ \ \overline{g^{A \dot{B}}} = g^{B \dot{A}} 
\\ \nonumber
\mathbf{HR}_{-}: \ \ && \ \ \overline{g^{A \dot{B}}} = -g^{B \dot{A}} 
\\ \nonumber
\mathbf{UR}: \ \ && \ \ \overline{g^{A \dot{B}}} = g^{A \dot{B}} 
\\ \nonumber
\mathbf{ER}_{+}: \ \ && \ \ \overline{g^{A \dot{B}}} = -g_{A \dot{B}} 
\\ \nonumber
\mathbf{ER}_{-}: \ \ && \ \ \overline{g^{A \dot{B}}} = g_{A \dot{B}} 
\\ \nonumber
\end{eqnarray}
where the overbar stands for the complex conjugation.

\begin{Definicja}
Hodge star $*$ is the linear map $* : \Lambda^{p}(\mathcal{M}) \rightarrow \Lambda^{4-p}(\mathcal{M})$ defined as follows
\begin{equation}
\label{definicja_gwiazdki_Hodga}
\ast \omega = \frac{1}{{p! \, (4-p)!}} \, \textrm{exp} \Big[ \frac{i \pi}{2}(p(4-p)-2) \Big]
 \in_{a_1 ... a_p b_1 ... b_{4-p}} \omega^{a_1 ... a_p} \,
 e^{b_1} \wedge ...\wedge e^{b_{4-p}}
\end{equation}
where $\omega \in \Lambda (\mathcal{M})  := \oplus_{p=0}^{4} \Lambda^p(\mathcal{M})$ is the $p$-form and it has the local representation
\begin{equation}
\omega = \frac{1}{p!} \omega_{a_1 ... a_p} e^{a_1} \wedge ... \wedge e^{a_p}
\end{equation}
\end{Definicja}
Operation $*$ defined by (\ref{definicja_gwiazdki_Hodga}) is the idempotent operation on $\Lambda (\mathcal{M})$: $\ast \ast \omega = \omega$. Under (\ref{definicja_gwiazdki_Hodga}) the $p$-forms defined by (\ref{formy_spinorowo}) behave as follows
\begin{eqnarray}
&& \ast 1 = vol \ , \ \ \ \ast g^{A \dot{B}} = \check{g}^{A\dot{B}} \ , \ \ \ \ast S^{AB} = S^{AB} \ , \ \ \ \ast S^{\dot{A} \dot{B}}  = -S^{\dot{A} \dot{B}} 
\\ \nonumber
&& \ast \check{g}^{A\dot{B}}  = g^{A \dot{B}} \ , \ \ \ \ast vol =1
\end{eqnarray}

Some identities involving (\ref{formy_spinorowo}) are helpful in calculations
\begin{eqnarray}
\label{identities}
&& g^{A \dot{B}} \wedge g^{C \dot{D}} = \in^{AC} S^{\dot{B} \dot{D}} + \in^{\dot{B} \dot{D}} S^{AC} 
\\ \nonumber
&& g^{A \dot{B}} \wedge S^{\dot{C}\dot{D}} = -\in^{\dot{B}\dot{D}} \check{g}^{A \dot{C}} - \in^{\dot{B}\dot{C}} \check{g}^{A \dot{D}} \ , \ \ \ 
S^{AB} \wedge g^{C \dot{D}}  = -\in^{AC} \check{g}^{B\dot{D}} - \in^{BC} \check{g}^{A \dot{D}}
\\ \nonumber
&&\frac{1}{4} S^{AB} \wedge S_{CD} = - \delta^{A}_{(C} \delta^{B}_{D)} \, vol \ , \ \ \ 
\frac{1}{4} S^{\dot{A}\dot{B}} \wedge S_{\dot{C}\dot{D}} = \delta^{\dot{A}}_{(\dot{C}} \delta^{\dot{B}}_{\dot{D})} \, vol \ , \ \ \  S^{AB} \wedge S_{\dot{C}\dot{D}} =0
\end{eqnarray}

The connection 1-forms $\Gamma_{ab}$ in null tetrad formalism are connected with the spinorial connection 1-forms $\mathbf{\Gamma}_{AB}$ and $\mathbf{\Gamma}_{\dot{A}\dot{B}}$ by the relations
\begin{equation}
\label{zwiazek_koneksji_spinorowej_i_klasycznej}
(\mathbf{\Gamma}_{AB}) = -\frac{1}{2}
                    \left[ \begin{array}{cc}
                            2\, \Gamma_{42} &  \Gamma_{12} + \Gamma_{34}  \\
                            \Gamma_{12} + \Gamma_{34} & 2\, \Gamma_{31}  
                            \end{array} \right]
                           \ , \ \ \ 
                           (\mathbf{\Gamma}_{\dot{A} \dot{B}}) = -\frac{1}{2}
                    \left[ \begin{array}{cc}
                            2\, \Gamma_{41} &  -\Gamma_{12} + \Gamma_{34}   \\
                             -\Gamma_{12} + \Gamma_{34} & 2\, \Gamma_{32}  
                            \end{array} \right] \ \ \ \ \ \ \ \ \ \ 
\end{equation}
The connection 1-forms can be decomposed according to
\begin{equation}
\label{rozklad_form_koneksji}
\mathbf{\Gamma}_{AB} = -\frac{1}{2} \mathbf{\Gamma}_{AB M \dot{N}} \, g^{M \dot{N}}
\ , \ \ \ 
\mathbf{\Gamma}_{\dot{A}\dot{B}} = -\frac{1}{2} \mathbf{\Gamma}_{\dot{A}\dot{B} M \dot{N}} \, g^{M \dot{N}}
\end{equation}
\begin{comment}
The explicit relations between $\mathbf{\Gamma}_{AB M \dot{N}}$, $\mathbf{\Gamma}_{\dot{A}\dot{B} M \dot{N}}$ and Ricci rotation coefficients $\Gamma_{abc}$ read
\begin{eqnarray}
\label{jawne_relacje_miedzy_koneksjami_obu_rodzajow}
&& \mathbf{\Gamma}_{11A\dot{B}} = \sqrt{2}
\left[\begin{array}{cc}
\Gamma_{424} & \Gamma_{422} \\
\Gamma_{421} & -\Gamma_{423}
\end{array}\right] \ , \ \ \ 
\mathbf{\Gamma}_{22A\dot{B}} = \sqrt{2}
\left[\begin{array}{cc}
\Gamma_{314} & \Gamma_{312} \\
\Gamma_{311} & -\Gamma_{313}
\end{array}\right]
\\ \nonumber
&& \mathbf{\Gamma}_{12A\dot{B}} = \frac{1}{\sqrt{2}}
\left[\begin{array}{cc}
\Gamma_{124} + \Gamma_{344} & \Gamma_{122} + \Gamma_{342} \\
\Gamma_{121} + \Gamma_{341} & -\Gamma_{123} - \Gamma_{343}
\end{array}\right]
\\ \nonumber
&& \mathbf{\Gamma}_{\dot{1}\dot{1} A\dot{B}} = \sqrt{2}
\left[\begin{array}{cc}
\Gamma_{414} & \Gamma_{412} \\
\Gamma_{411} & -\Gamma_{413}
\end{array}\right] \ , \ \ \ 
\mathbf{\Gamma}_{\dot{2}\dot{2} A\dot{B}} = \sqrt{2}
\left[\begin{array}{cc}
\Gamma_{324} & \Gamma_{322} \\
\Gamma_{321} & -\Gamma_{323}
\end{array}\right]
\\ \nonumber
&& \mathbf{\Gamma}_{\dot{1}\dot{2} A\dot{B}} = \frac{1}{\sqrt{2}}
\left[\begin{array}{cc}
-\Gamma_{124} + \Gamma_{344} & -\Gamma_{122} + \Gamma_{342} \\
-\Gamma_{121} + \Gamma_{341} & \Gamma_{123} - \Gamma_{343}
\end{array}\right]
\end{eqnarray}
where $A$ corresponds to the row index and $\dot{B}$ to the column index. 
\end{comment}
The general formula for the spinorial covariant derivative of the arbitrary spinor field $\Psi^{A \dot{B}}_{C \dot{D}}$ reads
\begin{eqnarray}
\label{covariant_derivative_for_spinors}
\nabla_{M \dot{N}} \Psi^{A \dot{B}}_{C \dot{D}} &=& \partial_{M \dot{N}} \Psi^{A \dot{B}}_{C \dot{D}} 
+ \mathbf{\Gamma}^{A}_{\ SM \dot{N}} \, \Psi^{S \dot{B}}_{C \dot{D}} 
- \mathbf{\Gamma}^{S}_{\ CM \dot{N}} \, \Psi^{A \dot{B}}_{S \dot{D}}
\\ \nonumber
&& \ \ \ \ \ \ \ \ \ \ + \mathbf{\Gamma}^{\dot{B}}_{\ \dot{S}M \dot{N}} \, \Psi^{A \dot{S}}_{C \dot{D}}
- \mathbf{\Gamma}^{\dot{S}}_{\ \dot{D}M \dot{N}} \, \Psi^{A \dot{B}}_{C \dot{S}}
\end{eqnarray}
where
\begin{equation}
\nabla_{A\dot{B}} := g^{a}_{\ A\dot{B}} \nabla_{a} \ , \ \ \ \partial_{A\dot{B}} := g^{a}_{\ A\dot{B}} \partial_{a} 
\end{equation}
and the matrices $g_{aA\dot{B}}$ are defined by the relation $g^{A\dot{B}} = g_{a}^{\ A\dot{B}}e^{a}$. Note that we have
\begin{equation}
\label{tozsamosci_macierzy_Pauliego}
g_{aA\dot{B}}g^{bA\dot{B}} = -2 \delta_{a}^{b} \ , \ \ \ g_{aA\dot{B}}g^{aC\dot{D}} = -2 \delta_{A}^{C} \delta_{\dot{B}}^{\dot{D}}
\end{equation}
The arbitrary vector $V$ has the form
\begin{equation}
\label{rozklad_wektora}
V = V^{a}\partial_{a} = -\frac{1}{2} V^{A\dot{B}} \partial_{A\dot{B}} \ , \ \ \ V^{a} = -\frac{1}{2} \, g^{a}_{\ A\dot{B}} V^{A\dot{B}} \ \Longleftrightarrow \ V^{A\dot{B}} = g_{a}^{\ A\dot{B}} V^{a}
\end{equation}

The first and second Cartan structure equations read
\begin{subequations}
\begin{eqnarray}
\label{pierwsze_rownania_struktury}
D g^{A \dot{B}} &=& d g^{A \dot{B}} + \Gamma^{A}_{\ \; C} \wedge g^{C \dot{B}} + 
                                  \Gamma^{\dot{B}}_{\ \; \dot{C}} \wedge g^{A \dot{C}} = 0
\\ 
\label{drugie_rownania_struktury}
R^{A}_{\ \; B} &=& d \Gamma^{A}_{\ \; B} + \Gamma^{A}_{\ \; C} \wedge \Gamma^{C}_{\ \; B}
\\ \nonumber
R^{\dot{A}}_{\ \; \dot{B}} &=& d \Gamma^{\dot{A}}_{\ \; \dot{B}} + \Gamma^{\dot{A}}_{\ \; \dot{C}} \wedge \Gamma^{\dot{C}}_{\ \; \dot{B}}
\end{eqnarray}
\end{subequations}
$R^{A}_{\ \; B}$ and $R^{\dot{A}}_{\ \; \dot{B}}$ are the components of curvature 2-forms of the connection $\Gamma^{A}_{\ \; B}$ or $\Gamma^{\dot{A}}_{\ \; \dot{B}}$, respectively, 
$D := - \frac{1}{2} \, g^{A \dot{B}} \, \nabla_{A \dot{B}}$ and $d := - \frac{1}{2} \, g^{A \dot{B}} \, \partial_{A \dot{B}}$. The 2-forms $R_{AB}= R_{(AB)}$ and $R_{\dot{A}\dot{B}} = R_{(\dot{A}\dot{B})}$ can be decomposed with respect to the bases $S^{AB}$ and $S^{\dot{A}\dot{B}}$ as follows
\begin{eqnarray}
\label{definicja_curvatury}
R_{AB} &=& - \frac{1}{2} \, C_{ABCD} \, S^{CD} + \frac{R}{24} \, S_{AB} 
           + \frac{1}{2} \, C_{AB \dot{C}\dot{D}} \, S^{\dot{C} \dot{D}}
\\ \nonumber
R_{\dot{A}\dot{B}} &=& - \frac{1}{2} \, C_{\dot{A}\dot{B}\dot{C}\dot{D}} \, S^{\dot{C} \dot{D}}
                  + \frac{R}{24} \, S_{\dot{A} \dot{B}} 
           + \frac{1}{2} \, C_{CD \dot{A}\dot{B}} \, S^{CD}
\end{eqnarray}
The objects $C_{ABCD}=C_{(ABCD)}$, $C_{\dot{A}\dot{B}\dot{C}\dot{D}}=C_{(\dot{A}\dot{B}\dot{C}\dot{D})}$, $C_{AB \dot{C}\dot{D}} = C_{(AB) \dot{C}\dot{D}} =C_{AB (\dot{C}\dot{D})}$ and $R$ have a transparent geometrical meaning. Namely, $C_{ABCD}$ is the spinorial image of the self-dual (SD) part of the Weyl tensor, $C_{\dot{A}\dot{B}\dot{C}\dot{D}}$ is the spinorial image of the anti-self-dual (ASD) part of the Weyl tensor, $C_{AB \dot{C}\dot{D}}$ is the spinorial image of the traceless Ricci tensor and, finally, $R$ is the curvature scalar.

Algebraic classification of totally symmetric 4-index spinors (like $C_{ABCD}$ and $C_{\dot{A}\dot{B}\dot{C}\dot{D}}$) has been presented in \cite{Penrose}. There are 6 different Petrov-Penrose types of such spinors in $\mathbf{CR}$ (namely $\textrm{[I]}$, $\textrm{[II]}$, $\textrm{[D]}$, $\textrm{[III]}$, $\textrm{[N]}$ and $[-]$). However, in $\mathbf{UR}$ there are 10 different types. These types are $[\textrm{I}_{r}]$, $[\textrm{I}_{rc}]$, $[\textrm{I}_{c}]$, $[\textrm{II}_{r}]$, $[\textrm{II}_{rc}]$, $[\textrm{D}_{r}]$, $[\textrm{D}_{c}]$, $[\textrm{III}_{r}]$, $[\textrm{N}_{r}]$ and $[-]$. In \cite{Rod_Hill_Nurowski} the authors used the following symbols for these types: $G_{r}$, $SG$, $G$, $II_{r}$, $II$, $D_{r}$, $D$, $III_{r}$, $N_{r}$ and $0$, respectively. We do not present here the details of this classification. They can be found in \cite{Rod_Hill_Nurowski, Chudecki_Classification}.

Note, that under the complex conjugation we have
\begin{eqnarray}
\nonumber
\mathbf{CR}: \ \ && \ \  \Gamma_{AB} \ \ , \ \ \Gamma_{\dot{A}\dot{B}} \ \ , \ \ 
C_{ABCD} \ \ , \ \ C_{\dot{A}\dot{B}\dot{C}\dot{D}} \ \ , \ \ C_{AB \dot{C}\dot{D}} \ \ \textrm{and} \ \ R \ \  \textrm{are holomorphic}
\\  \nonumber
\mathbf{HR}: \ \ && \ \ \Gamma_{\dot{A}\dot{B}}= \overline{\Gamma_{AB}}   \ \ , \ \
C_{\dot{A}\dot{B}\dot{C}\dot{D}} = \overline{C_{ABCD}} \ \ , \ \ 
         \overline{C_{AB \dot{C}\dot{D}}}= C_{CD \dot{A}\dot{B}} \ \ , \ \ 
    R  = \overline{R} 
\\ \nonumber
\mathbf{UR}: \ \ && \ \ \Gamma_{AB} \ \ , \ \ \Gamma_{\dot{A}\dot{B}} \ \ , \ \ 
C_{ABCD} \ \ , \ \ C_{\dot{A}\dot{B}\dot{C}\dot{D}} \ \ , \ \ C_{AB \dot{C}\dot{D}} \ \ \textrm{and} \ \ R \ \   \textrm{ are real}
\\ \nonumber
\mathbf{ER}: \ \ && \ \  \overline{\Gamma_{AB}} = \Gamma^{AB} \ \ , \ \ 
\overline{\Gamma_{\dot{A}\dot{B}}} = \Gamma^{\dot{A}\dot{B}} \ \ , \ \ 
\overline{C_{ABCD}} = C^{ABCD} \ \ , \ \ 
\overline{C_{\dot{A}\dot{B}\dot{C}\dot{D}}} = C^{\dot{A}\dot{B}\dot{C}\dot{D}} 
\\ \nonumber
&& \ \ \overline{C_{AB \dot{C}\dot{D}}} = C^{AB \dot{C}\dot{D}} \ \ , \ \ 
      R  = \overline{R} 
\end{eqnarray}
The Ricci identities for 1-index spinors read
\begin{eqnarray}
\label{Ricci_identities}
\frac{1}{2} \, \nabla^{E}_{\ \; ( \dot{C}} \nabla_{|E| \dot{D})} \, \Psi^{A} &=&
 \Psi^{E} \, C^{A}_{\ \; E \dot{C} \dot{D}}
\\ \nonumber
\frac{1}{2} \, \nabla_{(C}^{\ \ \; \dot{E}} \nabla_{D) \dot{E}} \, \Psi^{A} &=& \Psi^{E} \, \big( -C^{A}_{\ \; ECD} + \frac{R}{12} \, \in_{E(C}\delta^{A}_{\ \; D)} \big)
\\ \nonumber
\frac{1}{2} \, \nabla^{\ \ \; \dot{E}}_{( C} \nabla_{D) \dot{E}} \, \Psi^{\dot{A}} &=&
 \Psi^{\dot{E}} \, C^{\ \ \ \; \dot{A}}_{CD \ \ \dot{E}}
\\ \nonumber
\frac{1}{2} \, \nabla_{\ \; (\dot{C}}^{E} \nabla_{|E| \dot{D})} \, \Psi^{\dot{A}} &=& \Psi^{\dot{E}} \, \big( -C^{\dot{A}}_{\ \; \dot{E} \dot{C} \dot{D}} + \frac{R}{12} \, \in_{\dot{E} ( \dot{C}}\delta^{\dot{A}}_{\ \; \dot{D})} \big)
\end{eqnarray}
From (\ref{Ricci_identities}) we can easily obtain the Ricci identities for any spinor $\Psi^{A...\dot{B}...}_{C...\dot{D}...}$.

Let $m_{A}$ and $\mu_{A}$ be the pair of linearly independent spinors such that $m_{A} \mu^{A}=1$. SD curvature coefficients $C^{(i)}$, $i=1,...,5$ are defined as follows
\begin{eqnarray}
\label{krzywizna_rozlozona_na_wspolll}
2C_{ABCD} &=:& C^{(1)} \, \mu_{A}\mu_{B}\mu_{C}\mu_{D} + 4 C^{(2)} \, m_{(A}\mu_{B}\mu_{C}\mu_{D)} + 6 C^{(3)} \, m_{(A}m_{B}\mu_{C}\mu_{D)} 
\\ \nonumber
&& + 4 C^{(4)} \, m_{(A}m_{B}m_{C}\mu_{D)} +  C^{(5)} \, m_{A}m_{B}m_{C}m_{D} 
\end{eqnarray}

The relation between the traceless Ricci tensor $C_{a b}$ and its spinorial image $C_{AB\dot{C}\dot{D}}$ reads
\begin{equation}
C_{ab} = g_{a}^{\ A\dot{C}} g_{b}^{\ B\dot{D}} \, C_{AB\dot{C}\dot{D}} \ \Longleftrightarrow \ C_{AB\dot{C}\dot{D}} = \frac{1}{4} C_{ab} g^{a}_{\ A\dot{C}} g^{b}_{\ B\dot{D}}
\end{equation}
We decompose the traceless Ricci tensor $C_{AB \dot{C}\dot{D}}$ according to
\begin{equation}
\label{traceless_Ricci_tensor_two_congruences}
C_{AB\dot{M}\dot{N}} = m_{A}m_{B} A_{\dot{M}\dot{N}} + 2 m_{(A}\mu_{B)} B_{\dot{M}\dot{N}} + \mu_{A}\mu_{B} C_{\dot{M}\dot{N}}
\end{equation}
where $A_{\dot{M}\dot{N}}$, $B_{\dot{M}\dot{N}}$ and $C_{\dot{M}\dot{N}}$ are symmetric, 2-index dotted spinors (see \cite{Plebanski_Spinors} for more details). From those spinors the following scalars can be constructed
\begin{eqnarray}
\label{skroty_wielkosci_dla_Ricci}
&& a:= A_{\dot{A}\dot{B}}A^{\dot{A}\dot{B}} \ , \ \ \ b:= B_{\dot{A}\dot{B}}B^{\dot{A}\dot{B}}  \ , \ \ \ c := C_{\dot{A}\dot{B}}C^{\dot{A}\dot{B}} 
\\ \nonumber
&& r:= A_{\dot{A}\dot{B}}B^{\dot{A}\dot{B}}  \ , \ \ \ n:= A_{\dot{A}\dot{B}}C^{\dot{A}\dot{B}}  \ , \ \ \ s:= B_{\dot{A}\dot{B}}C^{\dot{A}\dot{B}} 
\end{eqnarray}
The equations for the eigenvectors and eigenvalues of the traceless Ricci tensor have the form (compare (\ref{tozsamosci_macierzy_Pauliego}) and (\ref{rozklad_wektora}))
\begin{equation}
\label{eigenvalues_equationss_forCab}
C^{a}_{\ b} V^{b} = \lambda \, V^{a} \ \ \Longleftrightarrow \ \ C_{AB\dot{C}\dot{D}} V^{B\dot{D}} = -\frac{1}{2} \lambda \, V_{A\dot{C}}
\end{equation}
The characteristic polynomial of the matrix $(C^{a}_{\ b})$ of the traceless Ricci tensor reads
\begin{equation}
\label{wielomian_charakterystyczny_traceless_Ricci_tensor}
 \mathcal{W} (x) := \det (C^{a}_{\ b} - x \delta^{a}_{\ b}) = \sum_{i=0}^{4} (-1)^{i} \, \underset{[i]}{\mathbb{C}} \, x^{4-i} \equiv \underset{[0]}{\mathbb{C}}x^4 - \underset{[1]}{\mathbb{C}} x^3+ \underset{[2]}{\mathbb{C}} x^2 - \underset{[3]}{\mathbb{C}} x + \underset{[4]}{\mathbb{C}}
\end{equation}
where coefficients $\underset{[i]}{\mathbb{C}}$ read
\begin{equation}
\underset{[0]}{\mathbb{C}} := 1 \ , \ \ \ \underset{[k]}{\mathbb{C}} := C^{a_{1}}_{\ \; [a_{1}} ... C^{a_{k}}_{\ \; a_{k}]} \ , \ \ \ k=1,2,3,4
\end{equation}
After long but elementary calculations we find the relations between $\underset{[i]}{\mathbb{C}}$ and the scalars given by (\ref{skroty_wielkosci_dla_Ricci})
\begin{eqnarray}
&&\underset{[0]}{\mathbb{C}} = 1 \ , \ \ \ \underset{[1]}{\mathbb{C}} = 0 \ , \ \ \ 
\underset{[2]}{\mathbb{C}} = 4 b -4  n  \ , \ \ \ \underset{[4]}{\mathbb{C}} = -16rs+ 8 nb + 4 a c+ 4 b^{2}
\\ \nonumber
&&\underset{[3]}{\mathbb{C}} = 16 A_{\dot{M}\dot{N}} C^{\dot{N}}_{\ \; \dot{R}}  B^{\dot{R}\dot{M}} \ \ \Longrightarrow \ \ \frac{1}{16^2} \underset{[3]}{\mathbb{C}}^{2} = rsn+\frac{1}{2}abc
-\frac{1}{2} cr^2-\frac{1}{2}bn^2-\frac{1}{2}as^2
\end{eqnarray}
Definitions of undotted and dotted \textsl{Plebański spinors} are \cite{Plebanski_Spinors, Plebanski_klasyfikacja_matter}
\begin{subequations}
\begin{eqnarray}
\label{Plebannnn_spinorrs}
V_{ABCD} &:=& 4 \, C_{(AB}^{\ \ \ \ \; \dot{M}\dot{N}} C_{AC)\dot{M}\dot{N}} = 4a \, m_{A}m_{B}m_{C}m_{D} + 16 r \,m_{(A}m_{B}m_{C}\mu_{D)}
\\
\nonumber
&& + (8n + 16 b) \, m_{(A}m_{B}\mu_{C}\mu_{D)} + 16 s \, m_{(A}\mu_{B}\mu_{C}\mu_{D)} + 4c\, \mu_{A}\mu_{B}\mu_{C}\mu_{D}
\\
V_{\dot{A}\dot{B}\dot{C}\dot{D}} &:=& 4 \, C_{MN(\dot{A}\dot{B}} C^{MN}_{\ \ \ \, \dot{C}\dot{D})}= 8 A_{(\dot{A}\dot{B}}C_{\dot{C}\dot{D})} - 8 B_{(\dot{A}\dot{B}}B_{\dot{C}\dot{D})}
\end{eqnarray}
\end{subequations}
Plebański spinors are 4-index and totally symmetric, so they can be classified analogously like SD and ASD Weyl spinors. 

The algebraic classification of the traceless Ricci tensor plays an important role in analysis presented in this paper. Such classification in 4-dimensional spaces with Lorentzian metric has been done in \cite{Plebanski_klasyfikacja_matter,Plebanski_Spinors} (15 types) and then it was generalized on the complex case in \cite{Przanowski_classification} (17 types). There are 33 different types of the traceless Ricci tensor in 4-dimensional real spaces equipped with the metric of signature $(++--)$. The detailed discussion which leads to this classification can be found in \cite{Chudecki_Classification}. Here we present only brief summary of the results of \cite{Chudecki_Classification}. The information about the algebraic type of the traceless Ricci tensor in $\mathbf{UR}$ is gathered in the following symbol
\begin{displaymath}
^{[\textrm{A}_{j}] \otimes [\textrm{B}_{k}]}[n_{1} E_{1} - n_{2} E_{2} - ...]^{v}_{(q_{1}q_{2}...)}
\end{displaymath}
Inside the square bracket all different eigenvalues $E_{i}$, $i=1,2,...,N_{0}$ of the polynomial $\mathcal{W} (x)$ together with their multiplicities $n_{i}$ are listed. Of course
\begin{eqnarray}
\nonumber
&&n_{1} + n_{2} + ... + n_{N_{0}} =4
\\ \nonumber
&&n_{1} E_{1} + n_{2}E_{2} + ... + n_{N_{0}}E_{N_{0}} =0
\end{eqnarray}
The last equality follows from the fact, that the matrix $(C^{a}_{\ b})$ is traceless. Complex eigenvalues are denoted by $Z$ and the real ones by $R$. Real eigenvalues have additional superscript which denotes the type of the corresponding eigenvector. $R^{s}$ means, that the eigenvector which corresponds to the eigenvalue $R$ is space-like, $R^{t}$ - time-like, $R^{n}$ - null, $R^{ns}$ - null or space-like, $R^{nt}$ - null or time-like and finally $R^{nst}$ means, that the eigenvector can be of the arbitrary type. [We use the same definition of space-like and time-like vectors, like in the Lorentzian case, i.e. $V^{a}V_{a} >0$ means, that $V^{a}$ is space-like, $V^{a}V_{a} <0$ stands for time-like vectors and finally, $V^{a}V_{a} =0$ means, that the vector is null]. Superscript $v$ denotes the number of eigenvectors and the numbers $q_{i}$ in round bracket describes the form of the minimal polynomial. Namely, the minimal polynomial of the matrix $(C^{a}_{\ b})$ has the form
\begin{equation}
\mathcal{W}_{\textrm{min}} (x) := \prod_{i=1}^{N_{0}} (x-E_{i})^{q_{i}}
\end{equation}
Finally, the symbol $[\textrm{A}_{j}] \otimes [\textrm{B}_{k}]$ describes Petrov-Penrose types of the Plebański spinors, $V_{ABCD}$ and $V_{\dot{A}\dot{B}\dot{C}\dot{D}}$, respectively. For example, $[\textrm{III}_{r}] \otimes [\textrm{N}_{r}]$ means, that $V_{ABCD}$ is of the type $[\textrm{III}_{r}]$ while $V_{\dot{A}\dot{B}\dot{C}\dot{D}}$ is of the type $[\textrm{N}_{r}]$.

%#####################################################################################

\renewcommand{\arraystretch}{1.5}
\setlength\arraycolsep{2pt}
\setcounter{equation}{0}

\section{Spaces equipped with one congruence of the SD null strings}
\label{Spaces_equipped_with_one_congruence_of_the_SD_null_strings}

\subsection{Definition and basic properties}

\begin{Definicja}
\label{definicja_foliacji_strun}
Congruence (foliation) of null strings in a complex (real) 4-dimensional manifold $\mathcal{M}$ is a family of totally null and totally geodesics 2-dimensional holomorphic (smooth) surfaces, such that for every point $p\in \mathcal{M}$ there exists only one surface of this family such that $p$ belongs to this surface.
\end{Definicja}
\begin{Twierdzenie}[Plebański, Rózga, \cite{Rozga}]
\label{Definicja_struny_zerowej}
The complex (real) 4-dimensional manifold $\mathcal{M}$ admits a congruence of the SD null strings, if there exists a holomorphic (smooth) 2-form $\Sigma$ and a 1-form $\sigma$ such that
\begin{subequations}
\begin{eqnarray}
\label{definicja_struny_condition_1}
 &&  \Sigma \wedge \Sigma =0
\\ 
\label{definicja_struny_condition_2}
 &&  * \Sigma = \Sigma
\\
\label{definicja_struny_condition_3}
 &&  d \Sigma = \sigma \wedge \Sigma
\end{eqnarray}
\end{subequations}
\hfill $\blacksquare$ 
\end{Twierdzenie}
In what follows we often abbreviate the cns defined by the 2-form $\Sigma$ by \textsl{$\Sigma$-congruence}. The 2-form $\Sigma$ which satisfies the conditions (\ref{definicja_struny_condition_1}) and (\ref{definicja_struny_condition_2}) from the Theorem \ref{Definicja_struny_zerowej} has the form $\Sigma=m_{A} m_{B} S^{AB}$ where $m_{A}$ is nowhere vanishing spinor. The condition (\ref{definicja_struny_condition_3}) implies, that 2-form $\Sigma$ is an element of the 2-surface completely integrable in the Frobenius sense. It gives restrictions on the spinor $m_{A}$. Indeed, one finds that spinor $m_{A}$ has to satisfy the following equations
\begin{equation}
\label{rownanie_struny_SD}
m^{A}m^{B} \nabla_{A \dot{M}} m_{B}=0
\end{equation}
The crucial equations (\ref{rownanie_struny_SD}) are called \textsl{the SD null string equations}. From (\ref{rownanie_struny_SD}) we find
\begin{equation}
\label{SD_null_strings_2}
\nabla_{A\dot{M}} m_{B} = Z_{A\dot{M}} m_{B} + \in_{AB} M_{\dot{M}}
\end{equation}
where $Z_{A\dot{M}}$ is the \textsl{Sommers vector} \cite{Rozga} and the spinor $M_{\dot{M}}$ is the \textsl{expansion of the congruence of the SD null strings} (see \cite{Rozga}). With fixed Riemannian structure the expansion describes the most important and invariant property of the cns. If $M_{\dot{M}}=0$ then the 2-dimensional distribution $\mathcal{D}_{m^{A}} := \{ m_{A} a_{\dot{B}}, m_{A} b_{\dot{B}} \}$, $a_{\dot{B}} b^{\dot{B}} \ne 0$ is parallely propagated. It means, that $\nabla_{X}V \in \mathcal{D}_{m^{A}}$ for every vector field $V \in \mathcal{D}_{m^{A}}$ and for arbitrary vector field $X$. Such cns are called \textsl{nonexpanding} or \textsl{plane}. If  $M_{\dot{M}} \ne 0$ then we deal with \textsl{expanding} (or \textsl{deviating}) cns. Under re-scaling $m^{A} = m m'^{A}$ the expansion and Sommers vector transform as follows
\begin{equation}
M'_{\dot{A}} = \frac{1}{m} \, M_{\dot{A}} \ , \ \ \ Z'_{A\dot{B}} = Z_{A\dot{B}} - \nabla_{A\dot{B}} \ln m
\end{equation}

\begin{comment}
The most plausible choice of function $m$ is probably the choice which makes $d\Sigma=0$ what is equivalent to the relation 
\begin{equation}
\label{kanoniczna_normalizacja}
2m^{A}Z_{A\dot{N}} = 3M_{\dot{N}}
\end{equation}
Using the terminology from \cite{Rozga} we say, that cns such that $d\Sigma=0$ is in \textsl{canonical normalization}. In further analysis we rarely use (\ref{kanoniczna_normalizacja}). However, in applications canonical normalization of the cns usually simplifies calculations.
\end{comment}

\subsection{SD Weyl spinor}

Consider the integrability conditions of the SD null strings equations. Acting on (\ref{SD_null_strings_2}) with $\nabla_{M \dot{N}}$ and using Ricci identities (\ref{Ricci_identities}), one arrives at the equations
\begin{subequations}
\begin{eqnarray}
\label{integrability_cond_1}
-2 m^{S} \, C_{CSMA} + \frac{R}{6} m_{(M} \in_{A)C} &=& m_{C} \, \nabla_{(M}^{\ \ \ \dot{N}}Z_{A) \dot{N}} 
- \in_{C(M} ( Z_{A)\dot{N}}M^{\dot{N}} - \nabla_{A)\dot{N}}M^{\dot{N}} ) \ \ \ \ \ \ \ \ \ \ 
\\ 
\label{integrability_cond_2}
2m^{S} \, C_{CS \dot{N}\dot{B}} &=& m_{C} \nabla^{A}_{\ \; (\dot{N}} Z_{|A|\dot{B})} + Z_{C(\dot{B}} M_{\dot{N})} - \nabla_{C(\dot{B}} M_{\dot{N})} \ \ \ \ \ 
\end{eqnarray}
\end{subequations}
Immediately we find
\begin{Twierdzenie}
\label{twierdzenie_o_spinorze_generujacym_wstege_i_spinorze_Penrosa}
If a spinor $m_{A}$ generates a congruence of SD null strings, then it is a Penrose spinor.
\end{Twierdzenie}
\begin{proof} Contracting the integrability conditions of the SD null strings equations (\ref{integrability_cond_1}) with $m^{C}m^{M}m^{A}$ we find $m^{S}m^{C}m^{M}m^{A} \, C_{CSMA}=0$.
\end{proof}
\begin{Twierdzenie}
\label{twierdzenie_o_spinorze_generujacym_wstege_nieekspandujaca_i_spinorze_Penrosa}
If a spinor $m_{A}$ generates a nonexpanding congruence of SD null strings then it is a multiple Penrose spinor.
\end{Twierdzenie}
\begin{proof} Contracting the integrability conditions of the nonexpanding ($M^{\dot{B}}=0$) SD null strings equations (\ref{integrability_cond_1}) with $m^{C}m^{M}$ we find $m^{S}m^{C}m^{M} \, C_{CSMA}=0$. 
\end{proof}

Define spinor $\mu^{A}$ by the relation $\mu^{A} m_{A}=1$. Then spinors $m_{A}$ and $\mu_{A}$ constitute the basis of the 1-index undotted spinors. The SD curvature coefficients $C^{(i)}$, $i=1,...,5$ are defined as usual by (\ref{krzywizna_rozlozona_na_wspolll}). Eqs. (\ref{integrability_cond_1}) constitute the set of 6 equations which relate SD Weyl spinor and the curvature scalar with the Sommers vector and the expansion of the cns. We easily conclude that 
\begin{equation}
\label{krzywizna_jedna_sstruna}
C^{(1)}=0 \ , \ \ \  \nabla_{(M}^{\ \ \ \dot{N}} Z_{A)\dot{N}} = C^{(4)} \, m_{A} m_{M} + 3C^{(3)} \, m_{(A}\mu_{M)} + 3C^{(2)} \, \mu_{A}\mu_{M}
\end{equation}
Note that the coefficient $C^{(5)}$ remains undetermined, so the existence of only one cns does not determine SD conformal curvature. Using (\ref{krzywizna_jedna_sstruna}), the integrability conditions (\ref{integrability_cond_1}) reduce to the equations
\begin{equation}
\label{definicja_Psi}
Z_{A\dot{B}}M^{\dot{B}} - \nabla_{A\dot{B}}M^{\dot{B}} = \frac{1}{6} (R-6C^{(3)} ) m_{A} - 2C^{(2)} \mu_{A} 
\end{equation}
From (\ref{definicja_Psi}) (with help of (\ref{krzywizna_jedna_sstruna})) we find the curvature scalar 
\begin{equation}
\label{skalar_krzywizny_R_jawnie}
\frac{3}{2} R = \mu_{A} (3 \nabla^{A\dot{N}} - Z^{A \dot{N}}) (2m^{M}Z_{M\dot{N}} - 3M_{\dot{N}}) + 2Z_{N\dot{N}}Z^{N\dot{N}}+ 3\nabla_{N\dot{N}}Z^{N\dot{N}}
\end{equation}
Finally (\ref{definicja_Psi}) leaves us with one constraint equation
\begin{equation}
\label{wiaz_na_obecnosc_jednej_struny}
\nabla^{M \dot{N}} \left( m_{M} (2m^{A}Z_{A\dot{N}} - 3M_{\dot{N}} ) \right) = 0
\end{equation}

Gathering: from 6 equations (\ref{integrability_cond_1}) we get the formulas (\ref{krzywizna_jedna_sstruna}) for curvature coefficients $C^{(1)}$, $C^{(2)}$, $C^{(3)}$ and $C^{(4)}$, the curvature scalar $R$ given by (\ref{skalar_krzywizny_R_jawnie}) and one constraint equation (\ref{wiaz_na_obecnosc_jednej_struny}). 

Petrov-Penrose types of $C_{ABCD}$ of the spaces equipped with one cns are gathered in the Table \ref{Tabela_Petrov_Penrose_types_via_properties_one_congruence}. In the case of the Einstein spaces the types $[\textrm{I}]$ (in $\mathbf{CR}$), $[\textrm{I}_{r}]$ and $[\textrm{I}_{rc}]$ (in $\mathbf{UR}$) are not admitted (compare \textsl{generalized Goldberg - Sachs Theorem} which has been formulated for the first time in \cite{Plebanski_surf}).

\begin{table}[ht]
\begin{center}
\begin{tabular}{|c|c|c|}   \hline
\multicolumn{3}{|c|}{Types in $\mathbf{CR}$}  \\ \hline
Curvature scalar & $M_{\dot{A}} \ne 0$ & $M_{\dot{A}}=0$   \\  \hline
$R\ne0$ & $[\textrm{I}]$, $[\textrm{II}]$, $[\textrm{D}]$, & $[\textrm{II}]$, $[\textrm{D}]$  \\ \cline{1-1} \cline{3-3}
$R=0$ & $[\textrm{III}]$, $[\textrm{N}]$, $[-]$  & $[\textrm{III}]$, $[\textrm{N}]$, $[-]$  \\ \hline
\multicolumn{3}{|c|}{Types in $\mathbf{UR}$}  \\ \hline
Curvature scalar & $M_{\dot{A}} \ne 0$ & $M_{\dot{A}}=0$   \\  \hline
$R\ne0$ & $[\textrm{I}_{r}]$, $[\textrm{I}_{rc}]$, $[\textrm{II}_{r}]$, $[\textrm{II}_{rc}]$,  & $[\textrm{II}_{r}]$, $[\textrm{II}_{rc}]$, $[\textrm{D}_{r}]$  \\ \cline{1-1} \cline{3-3}
$R=0$ & $[\textrm{D}_{r}]$, $[\textrm{III}_{r}]$, $[\textrm{N}_{r}]$, $[-]$ & $[\textrm{III}_{r}]$, $[\textrm{N}_{r}]$, $[-]$  \\ \hline
\end{tabular}
\caption{Petrov-Penrose types of the SD Weyl spinor of the spaces equipped with one congruence of SD null strings.}
\label{Tabela_Petrov_Penrose_types_via_properties_one_congruence}
\end{center}
\end{table}

\subsection{Traceless Ricci tensor}
\label{subsekcja_Analysis_of_algebraic_structure_of_traceless_Ricci_tensor_one_congr}

Now we find the possible types of the traceless Ricci tensor which are admitted by the space equipped with one cns in $\mathbf{CR}$ and $\mathbf{UR}$. The criteria described in \cite{Przanowski_classification} allow to establish the algebraic type of the traceless Ricci tensor in $\mathbf{CR}$. In \cite{Chudecki_Classification} we presented the similar criteria for the Plebański-Przanowski types of the traceless Ricci tensor in $\mathbf{UR}$.

If a space is equipped with one cns then the traceless Ricci tensor is determined with precision up to the factor $m^{A}C_{AB \dot{M}\dot{N}}$ and it is given by (\ref{integrability_cond_2}). Using the decomposition (\ref{traceless_Ricci_tensor_two_congruences}) we find $B_{\dot{A}\dot{B}}$ and $C_{\dot{A}\dot{B}}$ as
\begin{subequations}
\begin{eqnarray}
\label{definition_B_AB_1}
2B_{\dot{A}\dot{B}} &=& \mu_{N} ( Z^{N}_{\ \; (\dot{A}} M_{\dot{B})} -  \nabla^{N}_{\ (\dot{A}} M_{\dot{B})} ) + \nabla_{N (\dot{A}} Z^{N}_{\ \, \dot{B})}
\\
\label{definition_C_AB_1}
2C_{\dot{A}\dot{B}} &=& m_{N} ( \nabla^{N}_{\ ( \dot{A}} M_{\dot{B})} - Z^{N}_{\ \; ( \dot{A}} M_{\dot{B})} )
\end{eqnarray}
\end{subequations}
Spinor $A_{\dot{A}\dot{B}}$ is not determined and, consequently, the spaces equipped with one cns only admit all Plebański-Przanowski types of the traceless Ricci tensor. Stronger restrictions on the traceless Ricci tensor appear if spaces admit the nonexpanding cns. Indeed, one finds that $M_{\dot{A}}=0 \Longrightarrow C_{\dot{A}\dot{B}}=0 \Longrightarrow n=c=s=0$ what implies that the only nonzero scalars $\underset{[i]}{\mathbb{C}}$ are
\begin{equation}
\label{Riccitensor_uproszczoneskalary_1}
\underset{[2]}{\mathbb{C}} = 4 b \ , \ \ \ \underset{[4]}{\mathbb{C}} = 4 b^2
\end{equation}
The characteristic polynomial (\ref{wielomian_charakterystyczny_traceless_Ricci_tensor}) takes the form
\begin{equation}
\label{Wielomian_charakterystyczny_uproszczony_jednastruna}
\mathcal{W} (x) = x^{4}+4b \, x^{2} + 4b^{2} = (x^{2} + 2b)^{2}
\end{equation}
The Plebański spinors read
\begin{subequations}
\begin{eqnarray}
V_{ABCD} &=&  4 m_{(A}m_{B} \big( a \, m_{C}m_{D)} + 4 r \,m_{C}\mu_{D)}+  4 b \, \mu_{C}\mu_{D)} \big)
\\
V_{\dot{A}\dot{B}\dot{C}\dot{D}} &=& - 8 B_{(\dot{A}\dot{B}}B_{\dot{C}\dot{D})}
\end{eqnarray}
\end{subequations}
Using criteria given in \cite{Chudecki_Classification} and the form of the Plebański spinors we find that $\mathbf{UR}$ spaces equipped with one nonexpanding cns admit only 12 types of the traceless Ricci tensor (see Table \ref{Tabela_Ricci_via_one_nonexpanding_congruence_1_cong}). In the case of the $\mathbf{CR}$ there are only 9 different types of the traceless Ricci tensor.

There is an interesting connection between the existence of nonexpanding cns and the eigenvalues of the characteristic polynomial of the traceless Ricci tensor. Indeed, we have
\begin{Twierdzenie}[Przanowski, \cite{Przanowski}]
\label{twierdzenie_o_wartosciach_wlasnych}
In $\mathbf{CR}$ the existence of nonexpanding congruence of SD null strings implies that the characteristic polynomial of the traceless Ricci tensor has two double or one quadruple eigenvalue, namely $\lambda = \pm \sqrt{-2b}$.
\hfill $\blacksquare$ 
\end{Twierdzenie}
\begin{Wniosek}[Przanowski, \cite{Przanowski}]
\label{Wniosek_o_istnieniu_zerowego_wektora_wlasnego}
In $\mathbf{CR}$ the existence of nonexpanding congruence of SD null strings implies that the traceless Ricci tensor has at least one null eigenvector which is tangent to the null strings. Moreover:
\begin{itemize}
\item if there is exactly one null eigenvector then it is tangent to the null string
\item if there are two linearly independent null eigenvectors then at least one of them is tangent to the null string
\item if there are three or four linearly independent null eigenvectors then exactly two of them are tangent to the null string     \hfill $\blacksquare$ 
\end{itemize}
\end{Wniosek}

\textbf{Remark}. In $\mathbf{UR}$ Theorem \ref{twierdzenie_o_wartosciach_wlasnych} still holds true but Corollary \ref{Wniosek_o_istnieniu_zerowego_wektora_wlasnego} does not hold true anymore. Indeed, there are two algebraic types (namely $ ^{[\textrm{D}_{r}] \otimes [\textrm{D}_{c}]} [2Z - 2\bar{Z}]^{4}_{(11)}$ and $ ^{[\textrm{D}_{c}] \otimes [\textrm{D}_{c}]} [2R_{1}^{s}-2R_{2}^{t}]^{4}_{(11)}$) which are generated by the nonexpanding congruence of SD null strings, but they do not admit real null eigenvector of $C^{a}_{\ b}$. However, the following theorem holds true in both $\mathbf{CR}$ and $\mathbf{UR}$ (slightly weaker version of this theorem has been presented in \cite{Przanowski}).
\begin{Twierdzenie}
\label{Twierdzenie_nasze_o_zerowych_wektorach_wlasnych}
Let $\mu_{A}r_{\dot{B}}$ be a null eigenvector of the traceless Ricci tensor and let $m_{A}$ be any spinor such that $m_{A}\mu^{A} \ne 0$. Then $m_{A}r_{\dot{B}}$ is also a null eigenvector of the traceless Ricci tensor iff $C_{\dot{A}\dot{B}}r^{\dot{B}} =0$ where $C_{\dot{A}\dot{B}}$ is defined by $m_{A}$ and $\mu_{A}$ according to (\ref{traceless_Ricci_tensor_two_congruences}).
\end{Twierdzenie}
\begin{proof}
Without any loss of generality we put $\mu^{A}m_{A}=1$. Because the vector $\mu_{A}r_{\dot{A}}$ is null eigenvector of the traceless Ricci tensor then
\begin{eqnarray}
\label{pomocnicze_rownanie_na_wektor_wlasny_mu}
&&(m_{A}m_{B}A_{\dot{A}\dot{B}} + 2m_{(A} \mu_{B)} B_{\dot{A}\dot{B}} + \mu_{A}\mu_{B}C_{\dot{A}\dot{B}} ) \mu^{A}r^{\dot{A}} = -\frac{1}{2} \lambda \mu_{B} r_{\dot{B}}
\\ \nonumber
&& \Longrightarrow A_{\dot{A}\dot{B}} r^{\dot{A}}=0 \ , \ \ \ B_{\dot{A}\dot{B}} r^{\dot{A}}=-\frac{1}{2} \lambda r_{\dot{B}} \ , \ \ \  \lambda = \pm \sqrt{-2b}
\end{eqnarray}
From (\ref{pomocnicze_rownanie_na_wektor_wlasny_mu}) if follows that
\begin{equation}
C_{AB\dot{A}\dot{B}} m^{A} r^{\dot{A}} = \frac{1}{2} \lambda m_{B} r_{\dot{B}} \ \ \ \Longleftrightarrow \ \ \ C_{\dot{A}\dot{B}} r^{\dot{B}}=0
\end{equation}
\end{proof}
\textbf{Remark}. Obviously, Theorem \ref{Twierdzenie_nasze_o_zerowych_wektorach_wlasnych} holds true for the $C_{\dot{A}\dot{B}}=0$, i.e., for the nonexpanding cns and in this form it was presented in \cite{Przanowski}. However, for the proof it is enough to assume $C_{\dot{A}\dot{B}} r^{\dot{B}}=0 \ \Longrightarrow \ c=0$. It means, that Theorem \ref{Twierdzenie_nasze_o_zerowych_wektorach_wlasnych} holds true for special class of expanding cns such that $c=0$. 

In Table \ref{Tabela_Ricci_via_one_nonexpanding_congruence_1_cong} we gather information about the number of null eigenvectors of the traceless Ricci tensor in $\mathbf{UR}$ and $\mathbf{CR}$. $2^{TT}$ means, that both null eigenvectors are tangent to the null string and $2^{TN}$ means, that only one of them is tangent to the null string. In $\mathbf{CR}$ there is subtle difference between types $^{(2)} [4N]^{a}_{2}$ and $^{(2)} [4N]^{b}_{2}$ and this difference has not been recognized earlier in \cite{Przanowski}. 

\begin{table}[!ht]
\begin{center}
\begin{tabular}{|c|c|c|c|c|c|c|}   \hline
\multicolumn{3}{|c|}{Criteria}  & \multicolumn{2}{|c|}{Types in $\mathbf{UR}$}  & \multicolumn{2}{|c|}{Types in $\mathbf{CR}$} \\   \hline
 $b \ne 0$  & \multicolumn{2}{|c|}{ $ab \ne r^2$ } & $ ^{[\textrm{II}_{r}] \otimes [\textrm{D}_{r}]} [2R_{1}^{n}-2R_{2}^{n}]^{2}_{(22)}$ & $2^{TT}$ & $[2N_{1}-2N]^{b}_{4}$ & $2^{TT}$  \\ \cline{2-7}
   & \multicolumn{2}{|c|}{ $ab = r^2$, $b>0$ } & $ ^{[\textrm{D}_{r}] \otimes [\textrm{D}_{c}]} [2Z - 2\bar{Z}]^{4}_{(11)}$ & $0$ & $[2N_{1}-2N]_{2}$ & $4$ \\ \cline{2-7}
  & \multicolumn{2}{|c|}{ $ab=r^2$, $b<0$ } &  $ ^{[\textrm{D}_{c}] \otimes [\textrm{D}_{c}]} [2R_{1}^{s}-2R_{2}^{t}]^{4}_{(11)}$ & $0$ & $[2N_{1}-2N]_{2}$  & $4$ \\ 
  &  \multicolumn{2}{|c|}{  }  & $ ^{[\textrm{D}_{r}] \otimes [\textrm{D}_{r}]} [2R_{1}^{nst}-2R_{2}^{nst}]^{4}_{(11)}$ & $4$ & & \\ \cline{4-7}
  &  \multicolumn{2}{|c|}{  }  & $ ^{[\textrm{D}_{r}] \otimes [\textrm{D}_{r}]} [2R_{1}^{nst}-2R_{2}^{n}]^{3}_{(12)}$ & $3$ &  $[2N_{1}-2N]_{(1-2)}$ & $3$ \\ \hline
  $b = 0$ &  $B_{\dot{A}\dot{B}} \ne 0$ & $r \ne 0$  &$ ^{[\textrm{III}_{r}] \otimes [\textrm{N}_{r}]} [4R^{n}]^{1}_{(4)}$ & $1$ & $[4N]^{b}_{4}$ & $1$ \\ \cline{3-7}
  &  & $r=0$,  $a \ne 0 $ &  $ ^{[\textrm{N}_{r}] \otimes [\textrm{N}_{r}]} [4R^{nt}]^{2}_{(3)}$ & $1$ & $[4N]_{3}$ & $1$ \\ 
  &  &  &  $ ^{[\textrm{N}_{r}] \otimes [\textrm{N}_{r}]} [4R^{ns}]^{2}_{(3)}$ & $1$ & &  \\ \cline{3-7} 
  &  & $r=a=0$ &  $ ^{[-] \otimes [\textrm{N}_{r}]} [4R^{n}]^{2}_{(2)}$ & $2^{TN}$ & $^{(2)} [4N]^{a}_{2}$ & $2^{TN}$ \\ \cline{2-7} 
   &  $B_{\dot{A}\dot{B}} = 0$ & $r =0$, $a \ne 0$  & $ ^{[\textrm{N}_{r}] \otimes [-]} [4R^{n}]^{2}_{(2)}$ & $2^{TT}$ & $^{(2)} [4N]^{b}_{2}$ & $2^{TT}$ \\ \cline{3-7}
   &  &  $r=a=0$ & $ ^{[-] \otimes [-]} [4R^{nst}]^{4}_{(1)}$ & $4$ & $[4N]_{1}$ &  $4$ \\ 
   &  &    & $ ^{[-] \otimes [-]} [4R^{nst}]^{3}_{(2)}$ & $3$ & $^{(3)} [4N]_{2}$ & $3$ \\ \hline
\end{tabular}
\caption{Possible types of the traceless Ricci tensor in spaces equipped with one nonexpanding congruence of the SD null strings.}
\label{Tabela_Ricci_via_one_nonexpanding_congruence_1_cong}
\end{center}
\end{table}

\section{Spaces equipped with two congruences of the SD null strings}
\label{Sekcja_dwie_komplementarne_koonggruencje_strun}
\setcounter{equation}{0}

\subsection{Basic concepts and integrability conditions}

\begin{Definicja}
Two congruences of SD null strings $\Sigma$ and $\widetilde{\Sigma}$ are \textsl{distinct (complementary, transversal)}, if 
$\Sigma \wedge \widetilde{\Sigma} \ne 0$.
\end{Definicja}
\begin{Twierdzenie}
Let $\Sigma$-congruence be generated by the spinor $m_{A}$ and the $\widetilde{\Sigma}$-congruence be generated by the spinor $\mu_{A}$. These congruences are distinct, iff $m_{A}\mu^{A} \ne 0$.
\end{Twierdzenie}
\begin{proof} Because both congruences are self-dual, then $\Sigma=m_{A}m_{B} S^{AB}$ and $\widetilde{\Sigma}=\mu_{A}\mu_{B} S^{AB}$. Using (\ref{identities}) we find $\Sigma \wedge \widetilde{\Sigma} = 2 (m^{A}\mu_{A})^{2} vol$. Consequently $\Sigma \wedge \widetilde{\Sigma} \ne 0 \Longrightarrow m^{A}\mu_{A} \ne 0$. Conversely, consider two spinors $m_{A}$ and $\mu_{A}$ which generate cns such that $m_{A}\mu^{A} \ne 0$ . Then $ m_{A}\mu^{A} \ne 0 \Longrightarrow \Sigma \wedge \widetilde{\Sigma} \ne 0 $. 
\end{proof}

Spaces which admit two distinct congruences of the SD (or ASD) null strings are known as a \textsl{(complex or real) para-Hermite} spaces \cite{Flaherty, Przanowski_Formanski_Chudecki}. Spinors $m_{A}$ and $\mu_{A}$ can be normalized $m_{A} \mu^{A} =1$ without any loss of generality. We deal now with the equations for the $\Sigma$-congruence and $\widetilde{\Sigma}$-congruence
\begin{subequations}
\begin{eqnarray}
\label{dwie_struny_Sigma}
&& \nabla_{A\dot{B}} m_{C} = Z_{A\dot{B}} m_{C} + \in_{AC} M_{\dot{B}}
\\ 
\label{dwie_struny_Sigma_tilda}
&& \nabla_{A\dot{B}} \mu_{C} = \widetilde{Z}_{A\dot{B}} \mu_{C} + \in_{AC} \widetilde{M}_{\dot{B}}
\end{eqnarray}
\end{subequations}
and with the integrability conditions of the equations (\ref{dwie_struny_Sigma}) and (\ref{dwie_struny_Sigma_tilda})

\begin{subequations}
\begin{eqnarray}
\label{integrability_cond_dwie_struny_1}
-2 m^{S} \, C_{CSMA} + \frac{R}{6} m_{(M} \in_{A)C} &=& m_{C} \, \nabla_{(M}^{\ \ \ \dot{N}}Z_{A) \dot{N}} 
- \in_{C(M} ( Z_{A)\dot{N}}M^{\dot{N}} - \nabla_{A)\dot{N}}M^{\dot{N}} ) \ \ \ \ \ \ \ \ \ \ 
\\ 
\label{integrability_cond_dwie_struny_2}
-2 \mu^{S} \, C_{CSMA} + \frac{R}{6} \mu_{(M} \in_{A)C} &=& \mu_{C} \, \nabla_{(M}^{\ \ \ \dot{N}}\widetilde{Z}_{A) \dot{N}} - \in_{C(M} ( \widetilde{Z}_{A)\dot{N}}\widetilde{M}^{\dot{N}} - \nabla_{A)\dot{N}}\widetilde{M}^{\dot{N}} ) \ \ \ \ \ \ \ \ \ \ 
\\
\label{integrability_cond_dwie_struny_3}
2m^{S} \, C_{CS \dot{N}\dot{B}} &=& m_{C} \nabla^{A}_{\ \; (\dot{N}} Z_{|A|\dot{B})} + Z_{C(\dot{B}} M_{\dot{N})} - \nabla_{C(\dot{B}} M_{\dot{N})} \ \ \ \ \ 
\\
\label{integrability_cond_dwie_struny_4}
2\mu^{S} \, C_{CS \dot{N}\dot{B}} &=& \mu_{C} \nabla^{A}_{\ \; (\dot{N}} \widetilde{Z}_{|A|\dot{B})} + \widetilde{Z}_{C(\dot{B}} \widetilde{M}_{\dot{N})} - \nabla_{C(\dot{B}} \widetilde{M}_{\dot{N})} \ \ \ \ \ 
\end{eqnarray}
\end{subequations}
Because $m_{A}\mu^{A}=1 \ \Longrightarrow \ \nabla_{B\dot{B}} (m_{A}\mu^{A})=0$ we find the relation between the expansions and Sommers vectors for the $\Sigma$- and $\widetilde{\Sigma}$-congruences
\begin{equation}
\label{zwiazek_miedzy_Sommersami}
\widetilde{Z}_{A\dot{B}} + Z_{A\dot{B}} = m_{A} \widetilde{M}_{\dot{B}} - \mu_{A} M_{\dot{B}}
\end{equation}
In what follows we use the Eq. (\ref{zwiazek_miedzy_Sommersami}) to eliminate $\widetilde{Z}_{A\dot{B}}$.

Integrability conditions (\ref{integrability_cond_dwie_struny_1})-(\ref{integrability_cond_dwie_struny_2}) constitute the set of 12 equations. One quickly finds that all curvature coefficients $C^{(i)}$ are determined by
\begin{equation}
\label{SD_Weyl_spinor_two_congruences}
C^{(1)}=C^{(5)}=0 \ , \ \ \  \nabla_{(M}^{\ \ \ \dot{N}} Z_{A)\dot{N}} = C^{(4)} \, m_{A} m_{M} + 3C^{(3)} \, m_{(A}\mu_{M)} + 3C^{(2)} \, \mu_{A}\mu_{M}
\end{equation}
Eqs. (\ref{integrability_cond_dwie_struny_1}) and (\ref{integrability_cond_dwie_struny_2}) reduces to the following ones
\begin{subequations}
\label{dwie_complementarne_struny_obie}
\begin{eqnarray}
\label{dwie_complementarne_struny_pierwsza}
&&-2C^{(2)} \, \mu_{A} - \frac{6C^{(3)}-R}{6} \, m_{A} = Z_{A \dot{B}}M^{\dot{B}} - \nabla_{A \dot{B}}M^{\dot{B}}  
\\ 
\label{dwie_complementarne_struny_druga}
&&2C^{(4)} \, m_{A} + \frac{6C^{(3)}-R}{6} \, \mu_{A} = Z_{A \dot{B}} \widetilde{M}^{\dot{B}} + \nabla_{A \dot{B}}\widetilde{M}^{\dot{B}} + \mu_{A} M_{\dot{B}} \widetilde{M}^{\dot{B}}
\end{eqnarray}
\end{subequations}
Eqs. (\ref{dwie_complementarne_struny_pierwsza}) and (\ref{dwie_complementarne_struny_druga}) have to be consistent with each other and with (\ref{SD_Weyl_spinor_two_congruences}). After simple calculations we find that the curvature scalar $R$ has the form (\ref{skalar_krzywizny_R_jawnie}). We are left with three constraint equations which can be rearranged into the form
\begin{subequations}
\label{dwie_struny_wiaz_calosc}
\begin{eqnarray}
\label{dwie_struny_wiaz_1}
&& \nabla^{M \dot{N}} \left( m_{M} (2m^{A}Z_{A\dot{N}} - 3M_{\dot{N}} ) \right) = 0
\\
\label{dwie_struny_wiaz_2}
&& \nabla^{M \dot{N}} \left( \mu_{M} (2\mu^{A}Z_{A\dot{N}} + \widetilde{M}_{\dot{N}} ) \right) = 0
\\
\label{dwie_struny_wiaz_3}
&& \nabla^{M \dot{N}} ( m_{M} \widetilde{M}_{\dot{N}} + \mu_{M} M_{\dot{N}}  ) = 0
\end{eqnarray}
\end{subequations}
Gathering, from 12 equations (\ref{integrability_cond_dwie_struny_1}) and (\ref{integrability_cond_dwie_struny_2}) we obtained SD conformal curvature coefficients $C^{(i)}$, $i=1,...,5$ (5 equations), curvature scalar $R$ (1 equation) and the constraints (\ref{dwie_struny_wiaz_calosc}) (3 equations). The remaining 3 equations are identically satisfied. 

Integrability conditions (\ref{integrability_cond_dwie_struny_3})-(\ref{integrability_cond_dwie_struny_4}) give 12 equations for the 9 components of the traceless Ricci tensor. Three of these equations are identically satisfied and the rest bring us to the formulas 
\begin{subequations}
\label{definition_of_spinors_in_tracelessRicci}
\begin{eqnarray}
\label{definition_A_AB}
2A_{\dot{A}\dot{B}} &:=& \mu_{N} \nabla^{N}_{\ ( \dot{A}} \widetilde{M}_{\dot{B})} + \mu_{N} Z^{N}_{\ \; ( \dot{A}} \widetilde{M}_{\dot{B})} +  \widetilde{M}_{(\dot{A}} \widetilde{M}_{\dot{B})}
\\
\label{definition_B_AB}
2C_{\dot{A}\dot{B}} &:=& m_{N} \nabla^{N}_{\ ( \dot{A}} M_{\dot{B})} - m_{N} Z^{N}_{\ \; ( \dot{A}} M_{\dot{B})} 
\\
\label{definition_C_AB}
2B_{\dot{A}\dot{B}} &:=& \mu_{N} Z^{N}_{\ \; (\dot{A}} M_{\dot{B})} - \mu_{N} \nabla^{N}_{\ (\dot{A}} M_{\dot{B})} + \nabla_{N (\dot{A}} Z^{N}_{\ \, \dot{B})}
\end{eqnarray}
\end{subequations}
Formulas (\ref{definition_of_spinors_in_tracelessRicci}) prove important result. Expansions of congruences, the Sommers vector ($M_{\dot{A}}$, $\widetilde{M}_{\dot{B}}$, $Z_{A\dot{B}}$) and their covariant derivatives determine completely the form of the traceless Ricci tensor.

\subsection{SD Weyl spinor}

\begin{Twierdzenie}
\label{twierdzenie_o_nieistnieniu_dwoch_wsteg_wtypie_N}
If complex (real) space admits two distinct congruences of SD null strings then SD Weyl spinor $C_{ABCD}$ cannot be of the type $[\textrm{N}]$ ($[\textrm{N}_{r}]$).
\end{Twierdzenie}
\begin{proof} Assume, that SD Weyl spinor is of the type $[\textrm{N}]$ ($[\textrm{N}_{r}]$). Then there exists a nonzero complex (real) spinor $\alpha^{A}$ such that $C_{ABCD} \alpha^{A}=0$. Decomposing spinor $\alpha^{A}$ according to $\alpha^{A} = \alpha m^{A} + \beta \mu^{A}$, putting this into (\ref{krzywizna_rozlozona_na_wspolll}) and remembering that $C^{(1)}=C^{(5)}=0$ we find $\alpha C^{(i)}=0=\beta C^{(i)}$, $i=2,3,4$. All $C^{(i)}$ cannot be simultaneously equal zero, so the only solution is $\alpha=\beta=0$ what contradicts the assumption that spinor $\alpha^{A}$ in nonzero.
\end{proof}

Because $C^{(1)} = C^{(5)} = 0$ one arrives to the following formula for SD Weyl spinor
\begin{equation}
\label{krzywizna_rozlozona_na_wspolll_para_Hermite}
C_{ABCD} =:  m_{(A} \mu_{B} \left( 2 C^{(2)} \, \mu_{C}\mu_{D)} + 3 C^{(3)} \, m_{C}\mu_{D)} 
 + 2 C^{(4)} \, m_{C} m_{D)} \right)
\end{equation}
All possible types and forms of $C_{ABCD}$ are gathered in the Table \ref{Tabela_Petrov_Penrose_types} in which the following abbreviation has been used
\begin{equation}
\delta := 9 C^{(3)}C^{(3)} - 16 C^{(2)} C^{(4)}
\end{equation}

\begin{table}[ht]
\begin{center}
\begin{tabular}{|c|c|c|c|}   \hline
\multicolumn{2}{|c|}{Types} & Conditions  &  SD Weyl spinor  \\  \cline{1-2}
$\mathbf{CR}$ & $\mathbf{UR}$ &   &   \\  \hline
$\textrm{[I]}$ & $[\textrm{I}_{r}]$ &$C^{(2)} \ne 0$, $C^{(4)} \ne 0$, $\delta >0 $ &  $C_{ABCD}=2C^{(2)} m_{(A}\mu_{B}p^{+}_{C}p^{-}_{D)}$ \\ \cline{2-3}
& $[\textrm{I}_{rc}]$ & $C^{(2)} \ne 0$, $C^{(4)} \ne 0$, $\delta <0 $ &   $p^{\pm}_{A} := \mu_{A}  + \frac{1}{4C^{(2)}} \big( 3C^{(3)} \pm \sqrt{\delta} \big) m_{A}$ \\ \hline
$\textrm{[II]}$ & $[\textrm{II}_{r}]$ & $C^{(2)} = 0$, $C^{(3)} \ne 0$, $C^{(4)} \ne 0$ &  $C_{ABCD}= m_{(A}m_{B}\mu_{C} p_{D)}$ \\ 
 & & & $p_{A} := 3C^{(3)}\mu_{A} + 2 C^{(4)} m_{A}$ \\ \cline{3-4}
 & & $C^{(4)} = 0$, $C^{(2)} \ne 0$, $C^{(3)} \ne 0$ &  $C_{ABCD}= \mu_{(A}\mu_{B} m_{C} p_{D)} $ \\
&  & &  $p_{A} := 2C^{(2)}\mu_{A} + 3 C^{(3)} m_{A}$ \\ \cline{3-4}
 & & $C^{(2)} \ne 0$, $C^{(4)} \ne 0$, $\delta = 0$ &  $C_{ABCD}=2C^{(2)} m_{(A}\mu_{B}p_{C}p_{D)}$ \\ 
 &  &  & $p_{A} := \mu_{A} + \frac{3C^{(3)}}{4C^{(2)}} m_{A} $   \\ \hline
 $\textrm{[D]}$ & $[\textrm{D}_{r}]$ & $C^{(2)} =C^{(4)} = 0$, $C^{(3)} \ne 0$ &  $C_{ABCD}=3C^{(3)} m_{(A}m_{B}\mu_{C}\mu_{D)}$ \\ \hline
 $\textrm{[III]}$ &  $[\textrm{III}_{r}]$& $C^{(2)} = C^{(3)} = 0$, $C^{(4)} \ne 0$ &  $C_{ABCD}=2C^{(4)} m_{(A}m_{B}m_{C}\mu_{D)}$ \\ \cline{3-4}
 & & $C^{(3)} = C^{(4)} = 0$, $C^{(2)} \ne 0$ &  $C_{ABCD}=2C^{(2)} m_{(A}\mu_{B}\mu_{C}\mu_{D)}$ \\ \hline
 $[-]$ & $[-]$ & $C^{(2)} = C^{(3)} = C^{(4)} = 0$ &  $C_{ABCD}=0$ \\ \hline
\end{tabular}
\caption{Petrov-Penrose types of the SD Weyl spinor in the para-Hermite spaces.}
\label{Tabela_Petrov_Penrose_types}
\end{center}
\end{table}

Petrov-Penrose types of $C_{ABCD}$ in para-Hermite spaces depends strongly on expansions of the congruences (compare formulas (\ref{dwie_complementarne_struny_obie})). For example, if both congruences are nonexpanding, then $C^{(2)}=0 = C^{(4)}$ and $6C^{(3)}=R$. Consequently, $2C_{ABCD} = R \, m_{(A}m_{B} \mu_{C}\mu_{D)}$ so for $R\ne0$ we have type $[\textrm{D}]$ in $\mathbf{CR}$ or type $[\textrm{D}_{r}]$ in $\mathbf{UR}$ and for $R=0$ we have type $[-]$. In the case of Einstein spaces, the conditions $C_{ab}=0$ and $R=-4\Lambda$ reduce the possible Petrov-Penrose types of $C_{ABCD}$. In this case analysis is less straightforward (it is necessary to use the Bianchi identities) and we do not present the details here. Some portion of this analysis can be found in \cite{Przanowski_Formanski_Chudecki}. The results are gather in Tables \ref{Tabela_Petrov_Penrose_types_via_properties} and \ref{Tabela_Petrov_Penrose_types_via_properties_Einstein_case}.

\begin{table}[ht]
\begin{center}
\begin{tabular}{|c|c|c|c|}   \hline
\multicolumn{4}{|c|}{Types in $\mathbf{CR}$} \\ \hline
Curvature scalar & $M_{\dot{A}} \ne 0$, $\widetilde{M}_{\dot{A}} \ne 0$& $M_{\dot{A}}=0$, $\widetilde{M}_{\dot{A}} \ne 0$ &  $M_{\dot{A}}=0=\widetilde{M}_{\dot{A}}$  \\  \hline
$R\ne0$ & $[\textrm{I}]$, $[\textrm{II}]$, $[\textrm{D}]$, & $[\textrm{II}]$, $[\textrm{D}]$ &  $[\textrm{D}]$ \\ \cline{1-1} \cline{3-4}
$R=0$ & $[\textrm{III}]$, $[-]$ & $[\textrm{III}]$, $[-]$ &  $[-]$ \\ \hline
\multicolumn{4}{|c|}{Types in $\mathbf{UR}$} \\ \hline
Curvature scalar & $M_{\dot{A}} \ne 0$, $\widetilde{M}_{\dot{A}} \ne 0$& $M_{\dot{A}}=0$, $\widetilde{M}_{\dot{A}} \ne 0$ &  $M_{\dot{A}}=0=\widetilde{M}_{\dot{A}}$  \\  \hline
$R\ne0$ & $[\textrm{I}_{r}]$, $[\textrm{I}_{rc}]$, $[\textrm{II}_{r}]$, $[\textrm{D}_{r}]$, & $[\textrm{II}_{r}]$, $[\textrm{D}_{r}]$ &  $[\textrm{D}_{r}]$ \\ \cline{1-1} \cline{3-4}
$R=0$ & $[\textrm{III}_{r}]$, $[-]$ & $[\textrm{III}_{r}]$, $[-]$ &  $[-]$ \\ \hline
\end{tabular}
\caption{Petrov-Penrose types of SD Weyl spinor in para-Hermite spaces via properties of the congruences of the SD null strings.}
\label{Tabela_Petrov_Penrose_types_via_properties}
\end{center}
\end{table}

\begin{table}[ht]
\begin{center}
\begin{tabular}{|c|c|c|c|}   \hline
\multicolumn{4}{|c|}{Types in $\mathbf{CR}$} \\ \hline
Cosmological constant & $M_{\dot{A}} \ne 0$, $\widetilde{M}_{\dot{A}} \ne 0$& $M_{\dot{A}}=0$, $\widetilde{M}_{\dot{A}} \ne 0$ &  $M_{\dot{A}}=0=\widetilde{M}_{\dot{A}}$  \\  \hline
$\Lambda \ne 0$ & $[\textrm{D}]$, $[-]$ & not allowed &  $\textrm{[D]}$ \\ \cline{1-1} \cline{3-4}
$\Lambda =0 $ & & $[-]$ & $[-]$ \\ \hline
\multicolumn{4}{|c|}{Types in $\mathbf{UR}$} \\ \hline
Cosmological constant & $M_{\dot{A}} \ne 0$, $\widetilde{M}_{\dot{A}} \ne 0$& $M_{\dot{A}}=0$, $\widetilde{M}_{\dot{A}} \ne 0$ &  $M_{\dot{A}}=0=\widetilde{M}_{\dot{A}}$  \\  \hline
$\Lambda \ne 0$ & $[\textrm{D}_{r}]$, $[-]$ & not allowed &  $[\textrm{D}_{r}]$ \\ \cline{1-1} \cline{3-4}
$\Lambda =0 $ & & $[-]$ & $[-]$ \\ \hline
\end{tabular}
\caption{Petrov-Penrose types of SD Weyl spinor in Einstein para-Hermite spaces via properties of the congruences of the SD null strings.}
\label{Tabela_Petrov_Penrose_types_via_properties_Einstein_case}
\end{center}
\end{table}

\subsection{Traceless Ricci tensor}
\label{subsekcja_Analysis_of_algebraic_structure_of_traceless_Ricci_tensor}

The traceless Ricci tensor in para-Hermite spaces is completely determined via Sommers vector, expansions of both congruences and covariant derivatives of these objects (compare Eqs. (\ref{definition_of_spinors_in_tracelessRicci})). The most general case is when both congruences are expanding. It seems (somehow surprisingly) that in this case all Plebański-Przanowski types of the traceless Ricci tensor are admitted. The case with one nonexpanding congruence has been discussed in subsection \ref{subsekcja_Analysis_of_algebraic_structure_of_traceless_Ricci_tensor_one_congr} and this discussion remains valid in spaces equipped with one nonexpanding and one expanding cns. The last case concerns both congruences being nonexpanding. One finds that $M_{\dot{A}}=\widetilde{M}_{\dot{A}}=0$ implies $A_{\dot{A}\dot{B}}=C_{\dot{A}\dot{B}}=0$ and consequently $a=c=n=r=s=0$. Relations (\ref{Riccitensor_uproszczoneskalary_1}) and (\ref{Wielomian_charakterystyczny_uproszczony_jednastruna}) are still valid but Plebański spinors take the form
\begin{equation}
V_{ABCD} =  16b \, m_{(A}m_{B} \mu_{C}\mu_{D)} \ , \ \ \ 
V_{\dot{A}\dot{B}\dot{C}\dot{D}} = - 8 B_{(\dot{A}\dot{B}}B_{\dot{C}\dot{D})}
\end{equation}
However, these are not the only restrictions on the traceless Ricci tensor implied by the existence of two nonexpanding cns. We have transparent corollary which follows from the Theorem \ref{Twierdzenie_nasze_o_zerowych_wektorach_wlasnych}
\begin{Wniosek}
\label{wniosek_o_zerowych_wektorach_wlasnych_Ricci}
Assume, that the traceless Ricci tensor generated by two distinct, nonexpanding congruences of SD null strings admits a null eigenvector. Then
\begin{itemize}
\item there are two linearly independent null eigenvectors, the first is tangent to the $\Sigma$-congruence, the second to the $\widetilde{\Sigma}$-congruence
\item there are four linearly independent null eigenvectors, the first pair is tangent to the $\Sigma$-congruence, the second pair to the $\widetilde{\Sigma}$-congruence \hfill $\blacksquare$ 
\end{itemize}
\end{Wniosek}
Corollary \ref{wniosek_o_zerowych_wektorach_wlasnych_Ricci} proves, that the cases with one null eigenvector or three linearly independent null eigenvectors of the traceless Ricci tensor are not admitted in the spaces equipped with two nonexpanding cns. 

Using algebraic criteria given \cite{Chudecki_Classification}, the form of Plebański spinors and Corollary \ref{wniosek_o_zerowych_wektorach_wlasnych_Ricci} we arrive at the Table \ref{Tabela_Ricci_via_two_nonexpanding_congruence}.
\begin{table}[ht]
\begin{center}
\begin{tabular}{|c|c|c|c|}   \hline
\multicolumn{2}{|c|}{Criteria}  & Types in $\mathbf{UR}$ &  Types in $\mathbf{CR}$ \\  \hline
 \multicolumn{2}{|c|}{$b > 0$}   & $ ^{[\textrm{D}_{r}] \otimes [\textrm{D}_{c}]} [2Z-2\bar{Z}]^{4}_{(11)}$ & $[2N_{1}-2N]_{2}$ \\ \hline
 \multicolumn{2}{|c|}{$b < 0$}   & $ ^{[\textrm{D}_{c}] \otimes [\textrm{D}_{c}]} [2R_{1}^{s}-2R_{2}^{t}]^{4}_{(11)}$ & $[2N_{1}-2N]_{2}$ \\ 
 \multicolumn{2}{|c|}{  }  & $ ^{[\textrm{D}_{r}] \otimes [\textrm{D}_{r}]} [2R_{1}^{nst}-2R_{2}^{nst}]^{4}_{(11)}$ & \\ \hline
  $b = 0$ &  $B_{\dot{A}\dot{B}} \ne 0$ & $ ^{[-] \otimes [\textrm{N}_{r}]} [4R^{n}]^{2}_{(2)}$ & $^{(2)} [4N]^{a}_{2}$ \\ \cline{2-4} 
   &  $B_{\dot{A}\dot{B}} = 0$  & $ ^{[-] \otimes [-]} [4R^{nst}]^{4}_{(1)}$ & $[4N]_{1}$ \\ \hline
\end{tabular}
\caption{Possible types of the traceless Ricci tensor in para-Hermite spaces equipped with two nonexpanding congruence of the SD null strings.}
\label{Tabela_Ricci_via_two_nonexpanding_congruence}
\end{center}
\end{table}

\section{Spaces equipped with three or four congruences of SD null strings}
\setcounter{equation}{0}
\label{sekcja_o_dodatkowych_kongruencjach}

\subsection{General analysis}
\label{sekcja_o_dodatkowych_kongruencjach_general}

In this section we investigate the spaces which admit even richer structure then the para-Hermite spaces, i.e. we consider the spaces admitting three or four distinct congruences of SD null strings. As a starting point we take (complex or real) para-Hermite space equipped with two distinct congruences $\Sigma$ and $\widetilde{\Sigma}$ (generated by the spinors $m_{A}$ and $\mu_{A}$, respectively, such that $\mu^{A} m_{A}=1$). Consider now another cns, say  $\widehat{\Sigma}$-congruence, $\widehat{\Sigma} = s_{A}s_{B} S^{AB}$. Let $\widehat{\Sigma}$-congruence be complementary to the congruences $\Sigma$ and $\widetilde{\Sigma}$: $\Sigma \wedge \widehat{\Sigma} \ne 0$, $\widetilde{\Sigma} \wedge \widehat{\Sigma} \ne 0$. Decomposing the spinor $s^{A}$ according to the formula
\begin{equation}
\label{spinor_zA}
s^{A} = m m^{A} + \mu \mu^{A} = \mu \, (\xi m^{A} + \mu^{A}) \ , \ \ \ \xi := \frac{m}{\mu}
\end{equation}
we find that $m \ne 0$ and $\mu \ne 0$, otherwise the complementary conditions are not satisfied. As a consequence of existence of the $\widehat{\Sigma}$-congruence we find
\begin{equation}
\label{third_congruence_equations}
s^{A}s^{B} \nabla_{A\dot{B}} s_{B}=
s^{A} \left( \nabla_{A\dot{B}} \ln \xi + 2 Z_{A\dot{B}} - m_{A} \widetilde{M}_{\dot{B}} + \mu_{A} M_{\dot{B}} \right) = 0
\end{equation}
Acting on (\ref{third_congruence_equations}) with $s_{M} \nabla^{M\dot{B}}$ we find the integrability conditions of the equations (\ref{third_congruence_equations}) 
\begin{equation}
\label{warunki_calkowalnosci_trzeciejwsteggi}
2C^{(2)} \, \xi^2 - 3 C^{(3)} \, \xi + 2C^{(4)} = 0 
\end{equation}
According to Theorem \ref{twierdzenie_o_spinorze_generujacym_wstege_i_spinorze_Penrosa} the spinor which generates the cns must be proportional to a Penrose spinor. For the types $[\textrm{D}]$ and $[\textrm{III}]$ (in $\mathbf{CR}$) or types $[\textrm{D}_{r}]$ and $[\textrm{III}_{r}]$ (in $\mathbf{UR}$) we find, that Eq. (\ref{warunki_calkowalnosci_trzeciejwsteggi}) does not have any nonzero solutions (compare Table \ref{Tabela_Petrov_Penrose_types}). For all configurations of the $C^{(i)}$ coefficients which give the space of the type $[\textrm{II}]$ (in $\mathbf{CR}$) or $[\textrm{II}_{r}]$ (in $\mathbf{UR}$) we find that the equation (\ref{warunki_calkowalnosci_trzeciejwsteggi}) has exactly one nonzero solution. Indeed
\begin{eqnarray}
\label{solutions_for_x_typ_2}
&& \textrm{if } C^{(2)}=0 \textrm{ then } \xi=\frac{2C^{(4)}}{3C^{(3)}} \textrm{ and } s^{A} = \frac{\mu}{3 C^{(3)}} p^{A}
\\ \nonumber
&& \textrm{if } C^{(4)}=0 \textrm{ then } \xi=\frac{3C^{(3)}}{2C^{(2)}} \textrm{ and } s^{A} = \frac{\mu}{2 C^{(2)}} p^{A}
\\ \nonumber
&& \textrm{if } \delta =0 \textrm{ then } \xi=\frac{3C^{(3)}}{4C^{(2)}} \textrm{ and } s^{A}= \mu p^{A}
\end{eqnarray}
where $p^{A}$ is the third Penrose spinor (compare types $[\textrm{II}]$ and $[\textrm{II}_{r}]$ in the Table \ref{Tabela_Petrov_Penrose_types}). 

However, if $\delta \ne 0$ (in $\mathbf{CR}$) or $\delta > 0$ (in $\mathbf{UR}$) the equation (\ref{warunki_calkowalnosci_trzeciejwsteggi}) has two different nonzero solutions. One finds, that 
\begin{equation}
\label{solutions_for_x_typ_1}
\xi_{\pm}=\frac{1}{4C^{(2)}} \Big( 3C^{(3)} \pm \sqrt{ \delta } \Big)
\end{equation}
what implies $s_{A}=\mu p^{+}_{A}$ or $s_{A}=\mu p^{-}_{A}$ where spinors $p^{\pm}_{A}$ are the third and fourth Penrose spinor (compare types $[\textrm{I}]$ and $[\textrm{I}_{r}]$ in Table \ref{Tabela_Petrov_Penrose_types}). 

Summing up: the ratio $\xi$ can be found as a solution of the quadratic equation (\ref{warunki_calkowalnosci_trzeciejwsteggi}) and it depends on the curvature coefficients of the SD Weyl spinor $C^{(2)}$, $C^{(3)}$ and $C^{(4)}$. The space of the types $[\textrm{II}] \otimes [\textrm{any}]$ (in $\mathbf{CR}$) or $[\textrm{II}_{r}] \otimes [\textrm{any}]$ (in $\mathbf{UR}$) admits three distinct congruences of SD null strings, if
\begin{equation}
\label{third_congruence_equations_ostateczna_postac}
 (\xi m^{A} +  \mu^{A}) \left( \nabla_{A\dot{B}} \ln \xi + 2 Z_{A\dot{B}} - m_{A} \widetilde{M}_{\dot{B}} + \mu_{A} M_{\dot{B}} \right) = 0
\end{equation}
where $\xi$ is given by (\ref{solutions_for_x_typ_2}). However, for the spaces of the types $[\textrm{I}] \otimes [\textrm{any}]$ (in $\mathbf{CR}$) and $[\textrm{I}_{r}] \otimes [\textrm{any}]$ (in $\mathbf{UR}$) there are two solutions $\xi_{\pm}$. In this case we arrive at the equations
\begin{subequations}
\begin{eqnarray}
\label{cztery_kongruencje_1}
&&(\xi_{+} m^{A} +  \mu^{A}) \left( \nabla_{A\dot{B}} \ln \xi_{+} + 2 Z_{A\dot{B}} - m_{A} \widetilde{M}_{\dot{B}} + \mu_{A} M_{\dot{B}} \right) = 0
\\
\label{cztery_kongruencje_2}
&&(\xi_{-} m^{A} +  \mu^{A}) \left( \nabla_{A\dot{B}} \ln \xi_{-} + 2 Z_{A\dot{B}} - m_{A} \widetilde{M}_{\dot{B}} + \mu_{A} M_{\dot{B}} \right) = 0
\end{eqnarray}
\end{subequations}
If (\ref{cztery_kongruencje_1}) are satisfied and (\ref{cztery_kongruencje_2}) are not (or vice versa), then the space of the types $[\textrm{I}] \otimes [\textrm{any}]$ (in $\mathbf{CR}$) or $[\textrm{I}_{r}] \otimes [\textrm{any}]$ (in $\mathbf{UR}$) admits three distinct congruences of SD null strings. But if both (\ref{cztery_kongruencje_1}) and (\ref{cztery_kongruencje_2}) are satisfied, there are four distinct congruences of SD null strings.

[\textbf{Remark}. Equation (\ref{warunki_calkowalnosci_trzeciejwsteggi}) has a solution for arbitrary $\xi$ only if $C^{(i)}=0, i=2,3,4$. It means that arbitrary spinor $s^{A}$ generates the cns only if the spaces are left conformally flat. Such spaces admit infinitely many cns.]

We gather the results of this section in two Theorems.

\begin{Twierdzenie} 
\label{twierdzenie_o_typie_IxI}
If a complex (real) space with $C_{ABCD} \ne 0$ admits four distinct congruences of SD null strings, then
\begin{eqnarray}
\nonumber
&& (i) \ \ \  \textrm{the space is of the type } [\textrm{I}] \otimes [\textrm{any}] \ ([\textrm{I}_{r}] \otimes [\textrm{any}])
\\ \nonumber
&& (ii) \ \ \  \textrm{all congruences are necessarily expanding}
\end{eqnarray}
\end{Twierdzenie}
\begin{proof} 
Spinors which generate distinct cns are mutually linearly independent. The existence of four distinct cns implies the existence of four mutually linearly independent spinors and all of them are Penrose spinors (compare Theorem \ref{twierdzenie_o_spinorze_generujacym_wstege_i_spinorze_Penrosa}). Consequently the space must be of the type $[\textrm{I}] \otimes [\textrm{any}]$ ($[\textrm{I}_{r}] \otimes [\textrm{any}]$). It proves $(i)$. According to Theorem \ref{twierdzenie_o_spinorze_generujacym_wstege_nieekspandujaca_i_spinorze_Penrosa} the existence of nonexpanding cns implies that the space is algebraically degenerated what contradicts $(i)$. It proves $(ii)$.
\end{proof}
\begin{Twierdzenie} 
\label{Twierdzenie_trzy_struny_jedna_niie}
If a complex (real) space with $C_{ABCD} \ne 0$ admits three distinct congruences of SD null strings and one of them is nonexpanding, then
\begin{eqnarray}
\nonumber
&& (i) \ \ \  \textrm{the space is of the type } [\textrm{II}] \otimes [\textrm{any}] \ ([\textrm{II}_{r}] \otimes [\textrm{any}])
\\ \nonumber
&& (ii) \ \ \  \textrm{curvature scalar R is necessarily nonzero } R \ne 0
\end{eqnarray}
\end{Twierdzenie}
\begin{proof}
Nonexpanding cns is generated by the multiple Penrose spinor (compare Theorem \ref{twierdzenie_o_spinorze_generujacym_wstege_nieekspandujaca_i_spinorze_Penrosa}). Moreover, all spinors which generate distinct cns are mutually linearly independent and they are Penrose spinors (compare Theorem \ref{twierdzenie_o_spinorze_generujacym_wstege_i_spinorze_Penrosa}). The existence of three distinct cns (one of them being nonexpanding) implies the existence of three Penrose spinors (one of them is double). Consequently the space must be of the type $[\textrm{II}] \otimes [\textrm{any}]$ ($[\textrm{II}_{r}] \otimes [\textrm{any}]$). It proves $(i)$. $(ii)$ follows immediately from Theorem \ref{twierdzenie_o_spinorze_generujacym_wstege_nieekspandujaca_i_spinorze_Penrosa}. 
\end{proof}

Possible Petrov-Penrose types of the spaces equipped with three complementary cns are listed in the Table \ref{Tabela_Petrov_Penrose_types_via_properties_threee_congruencces}.
\begin{table}[ht]
\begin{center}
\begin{tabular}{|c|c|c|}   \hline
\multicolumn{3}{|c|}{Types in $\mathbf{CR}$} \\ \hline
Curvature scalar & $M_{\dot{A}} \ne 0$, $\widetilde{M}_{\dot{A}} \ne 0$, $\widehat{M}_{\dot{A}} \ne 0$ &  $M_{\dot{A}} = 0$, $\widetilde{M}_{\dot{A}} \ne 0$, $\widehat{M}_{\dot{A}} \ne 0$  \\  \hline
$R\ne0$ & $[\textrm{I}]$, $[\textrm{II}]$, $[-]$ &  $[\textrm{II}]$ \\ \cline{1-1} \cline{3-3}
$R=0$ &  &  $[-]$ \\ \hline
\multicolumn{3}{|c|}{Types in $\mathbf{UR}$} \\ \hline
Curvature scalar & $M_{\dot{A}} \ne 0$, $\widetilde{M}_{\dot{A}} \ne 0$, $\widehat{M}_{\dot{A}} \ne 0$ &  $M_{\dot{A}} = 0$, $\widetilde{M}_{\dot{A}} \ne 0$, $\widehat{M}_{\dot{A}} \ne 0$  \\  \hline
$R\ne0$ & $[\textrm{I}_{r}]$, $[\textrm{II}_{r}]$, $[-]$ &  $[\textrm{II}_{r}]$ \\ \cline{1-1} \cline{3-3}
$R=0$ &  &  $[-]$ \\ \hline
\end{tabular}
\caption{Petrov-Penrose types of SD Weyl spinor in spaces equipped with three distinct congruences of SD null strings via properties of these congruences.}
\label{Tabela_Petrov_Penrose_types_via_properties_threee_congruencces}
\end{center}
\end{table}

\subsection{The metric of the spaces equipped with three distinct congruences of SD null strings}
\label{subsekcja_o_wielu_kongrueencjach}

The general metric of the space equipped with three distinct cns has been found in \cite{Rozga_Robinson}. The authors considered a para-Hermite space equipped with the congruences $\Sigma$ and $\widetilde{\Sigma}$, given by the Pfaff systems $dz^{A}=0$ and $dz^{\dot{A}}=0$, respectively. It is well known \cite{Flaherty} that the metric of such a space can be brought to the form $ds^{2} = 2f_{A\dot{B}} \, dz^{A} dz^{\dot{B}}$ where $f_{A\dot{B}}=f_{A\dot{B}} (z^{M}, z^{\dot{M}})$ are arbitrary complex holomorphic (real smooth) functions. The third congruence, $\widehat{\Sigma}$ implies the existence of the functions $f^{\dot{A}}$ such that this congruence is defined by the Pfaff system $df^{\dot{A}}=0$. It means that the alternate form of the metric reads $ds^2 = 2 \omega_{A\dot{B}} \, dz^{A} df^{\dot{B}}$, where $\omega_{A\dot{B}}=\omega_{A\dot{B}} (z^{M}, f^{\dot{M}})$. Of course
\begin{equation}
\label{para_hermite_metric_twoforrms}
ds^{2}=2f_{A\dot{B}} \, dz^{A} dz^{\dot{B}} = 2\omega_{A\dot{B}} \, dz^{A} df^{\dot{B}}
\end{equation}
Treating $f^{\dot{A}}$ as a functions of $(z^{A}, z^{\dot{B}})$ we arrive at the equations
\begin{subequations}
\begin{eqnarray}
\label{definicja_efa}
&& f_{A\dot{B}} = \omega_{A\dot{M}} \, \frac{\partial f^{\dot{M}}}{\partial z^{\dot{B}}}
\\
\label{rownanie_na_trzecia_konggggruencje}
&& \frac{\partial f^{\dot{B}}}{\partial z^{(M}} \, \omega_{A)\dot{B}}=0
\end{eqnarray}
\end{subequations}
Eqs. (\ref{rownanie_na_trzecia_konggggruencje}) can be easily solved
\begin{equation}
\label{solution_for_omega}
\omega_{A\dot{B}} = \Omega \, \frac{\partial f_{\dot{B}}}{\partial z^{A}}
\end{equation}
where $\Omega = \Omega (z^{A}, z^{\dot{B}})$ is arbitrary function. Putting (\ref{solution_for_omega}) into (\ref{definicja_efa}) we arrive at the metric \cite{Rozga}
\begin{equation}
\label{metryka_z_trzema_kongruencjami_SDstrun}
ds^{2} = 2 \Omega \, \frac{\partial f_{\dot{M}}}{\partial z^{A}} \frac{\partial f^{\dot{M}}}{\partial z^{\dot{B}}} \, dz^{A} dz^{\dot{B}}
\end{equation}
Generally, the metric (\ref{metryka_z_trzema_kongruencjami_SDstrun}) admits three distinct expanding cns and its SD Weyl spinor is of the types $[\textrm{I}]$ or $[\textrm{II}]$ in $\mathbf{CR}$ ($[\textrm{I}_{r}]$ or $[\textrm{II}_{r}]$ in $\mathbf{UR}$). This metric cannot be Einsteinian. The metric depends on three functions $\Omega$ and $f^{\dot{B}}$ of the four variables $(z^{A}, z^{\dot{B}})$. Note, that all congruences are complementary to each other so $\Sigma \wedge \widetilde{\Sigma} \ne 0$, $\Sigma \wedge \widehat{\Sigma} \ne 0$ and $\widetilde{\Sigma} \wedge \widehat{\Sigma} \ne 0$. It follows that 
\begin{equation}
\label{warunki_na_kongruencje}
 \frac{\partial f^{\dot{A}}}{\partial z^{B}} \frac{\partial f_{\dot{A}}}{\partial z_{B}} \ne 0 \ , \ \ \ 
 \frac{\partial f^{\dot{A}}}{\partial z^{\dot{B}}} \frac{\partial f_{\dot{A}}}{\partial z_{\dot{B}}} \ne 0 
\end{equation}

The special case with one of these congruences being nonexpanding has not been considered in \cite{Rozga_Robinson}. In this case we arrive at the following theorem:
\begin{Twierdzenie} If a complex (real) space admits three distinct congruences of SD null strings and one of them is nonexpanding, then the metric of this space can be brought to the form
\begin{equation}
\label{metryka_z_trzema_kongruencjami_SDstrun_jedna_nonexpanding}
ds^{2} = 2 \Omega \, \frac{\partial f_{\dot{M}}}{\partial z^{A}} \frac{\partial f^{\dot{M}}}{\partial z^{\dot{B}}} \, dz^{A} dz^{\dot{B}}
\end{equation}
where $\Omega=\Omega(z^{\dot{A}})$, $f^{\dot{M}}=f^{\dot{M}} (z^{A}, z^{\dot{B}})$ and $(z^{A}, z^{\dot{B}})$ are local coordinates.
\end{Twierdzenie}
\begin{proof} The space admits three distinct cns, so the metric has the form (\ref{metryka_z_trzema_kongruencjami_SDstrun}). Detailed calculations prove, that the expansion of the $\widehat{\Sigma}$-congruence is proportional to the $\frac{\partial \omega_{M\dot{N}}}{\partial z_{M}}$. Consider now $\widehat{\Sigma}$-congruence as a nonexpanding one. We have
\begin{equation}
\label{warunek_na_brak_ekspansji_Sigmahat_kongruencji}
\frac{\partial \omega_{M\dot{N}}}{\partial z_{M}}=0
\end{equation}
Feeding (\ref{warunek_na_brak_ekspansji_Sigmahat_kongruencji}) with (\ref{solution_for_omega}) one obtains
\begin{equation}
\label{war_brak_konggg}
\frac{\partial \Omega}{\partial z_{A}} \, \frac{\partial f_{\dot{B}}}{\partial z^{A}}=0
\end{equation}
Contracting (\ref{war_brak_konggg}) with $\frac{\partial f^{\dot{B}}}{\partial z^{M}}$ we find
\begin{equation}
\frac{\partial \Omega}{\partial z^{M}} \, \frac{\partial f^{\dot{A}}}{\partial z^{B}} \frac{\partial f_{\dot{A}}}{\partial z_{B}} =0
\end{equation}
Finally, because of the (\ref{warunki_na_kongruencje}) one arrives at the formula
\begin{equation}
\frac{\partial \Omega}{\partial z^{M}} = 0 \ \Longleftrightarrow \ \Omega=\Omega (z^{\dot{N}})
\end{equation}
\end{proof}

\subsection{The metric of the spaces equipped with four distinct congruences of SD null strings}
\label{subsekcja_o_czterech_kongrueencjach}

Four distinct cns are admitted only by the spaces with algebraically general SD Weyl spinor. We assume the existence of congruences $\Sigma$, $\widetilde{\Sigma}$, $\widehat{\Sigma}$ and $\widehat{\widehat{\Sigma}}$ given by the Pfaff systems $dz^{A}=0$, $dz^{\dot{A}}=0$, $df^{A}=0$ and $df^{\dot{A}}=0$, respectively. All those congruences are complementary what implies 
\begin{equation}
\label{warunki_na_komplementarnosc_czterech_strun}
\Sigma \wedge \widetilde{\Sigma} \ne 0 \ , \ \ \ \Sigma \wedge \widehat{\Sigma} \ne 0 \ , \ \ \ 
\Sigma \wedge \widehat{\widehat{\Sigma}} \ne 0 \ , \ \ \  \widetilde{\Sigma} \wedge \widehat{\Sigma} \ne 0 \ , \ \ \ 
\widetilde{\Sigma} \wedge \widehat{\widehat{\Sigma}} \ne 0 \ , \ \ \ 
\widehat{\Sigma} \wedge \widehat{\widehat{\Sigma}} \ne 0 
\end{equation}
Writing $f^{A} = f^{A} (z^{B}, z^{\dot{C}})$ and $f^{\dot{A}} = f^{\dot{A}} (z^{B}, z^{\dot{C}})$, from (\ref{warunki_na_komplementarnosc_czterech_strun}) we find that
\begin{eqnarray}
&& \alpha_{1} := \frac{\partial f^{M}}{\partial z^{N}} \frac{\partial f_{M}}{\partial z_{N}} \ne 0 \ , \ \ \ 
 \alpha_{2} := \frac{\partial f^{M}}{\partial z^{\dot{N}}} \frac{\partial f_{M}}{\partial z_{\dot{N}}} \ne 0 \ , \ \ \
 \alpha_{3} := \frac{\partial f^{\dot{M}}}{\partial z^{N}} \frac{\partial f_{\dot{M}}}{\partial z_{N}} \ne 0
 \\ 
 && \alpha_{4} := \frac{\partial f^{\dot{M}}}{\partial z^{\dot{N}}} \frac{\partial f_{\dot{M}}}{\partial z_{\dot{N}}} \ne 0 \ , \ \ \ 
 \alpha_{1} \alpha_{4} + \alpha_{2} \alpha_{3} - 4 \, \frac{\partial f^{A}}{\partial z^{M}} \frac{\partial f_{A}}{\partial z^{\dot{N}}} \frac{\partial f^{\dot{B}}}{\partial z_{M}} \frac{\partial f_{\dot{B}}}{\partial z_{\dot{N}}} \ne 0
\end{eqnarray}
Then we arrive at the following 
\begin{Twierdzenie} If a complex (real) space admits four distinct congruences of SD null strings, then the metric of this space can be brought to the form
\begin{equation}
\label{metryka_przestrzeni_z_czterema_kongruencjami_strun}
ds^{2} = 2\Omega \, \left( \alpha_{3} \, \frac{\partial f^{M}}{\partial z^{A}} \frac{\partial f_{M}}{\partial z^{\dot{B}}} - \alpha_{1} \, \frac{\partial f^{\dot{M}}}{\partial z^{A}} \frac{\partial f_{\dot{M}}}{\partial z^{\dot{B}}} \right) dz^{A} dz^{\dot{B}}
\end{equation}
where $(z^{A}, z^{\dot{B}})$ are local coordinates, $\Omega=\Omega(z^{A},z^{\dot{B}})$ and the functions $f^{M}=f^{M} (z^{A}, z^{\dot{B}})$ and $f^{\dot{M}}=f^{\dot{M}} (z^{A}, z^{\dot{B}})$ satisfy the set of three equations
\begin{equation}
\label{rownania_jakie_zostaly_dla_metryki_cztery_kongr}
\frac{\partial f_{A}}{\partial z^{M}} \frac{\partial f_{\dot{B}}}{\partial z_{M}} 
\frac{\partial f^{A}}{\partial z^{(\dot{R}}} \frac{\partial f^{\dot{B}}}{\partial z^{\dot{S})}} = 0
\end{equation}
\end{Twierdzenie}
\begin{proof}
Consider the congruences $\Sigma$ and $\widetilde{\Sigma}$. Analogously to (\ref{para_hermite_metric_twoforrms}) the metric reads $ds^{2} = 2f_{A\dot{B}} \, dz^{A} dz^{\dot{B}}$. Alternatively, choosing $\widehat{\Sigma}$ and $\widehat{\widehat{\Sigma}}$-congruences as a basic ones, we have $ds^{2}=2\omega_{A\dot{B}} \, df^{A} df^{\dot{B}}$. Of course
\begin{equation}
2f_{A\dot{B}} \, dz^{A} dz^{\dot{B}} = 2\omega_{A\dot{B}} \, df^{A} df^{\dot{B}}
\end{equation}
what leads to the equations
\begin{subequations}
\begin{eqnarray}
\label{cztery_kongruencje_row1}
&& f_{A\dot{B}} = \omega_{M\dot{N}} \left( \frac{\partial f^{M}}{\partial z^{A}} \frac{\partial f^{\dot{N}}}{\partial z^{\dot{B}}} + \frac{\partial f^{M}}{\partial z^{\dot{B}}} \frac{\partial f^{\dot{N}}}{\partial z^{A}} \right)
\\
\label{cztery_kongruencje_row2}
&& \omega_{M\dot{N}} \, \frac{\partial f^{M}}{\partial z^{(A}} \frac{\partial f^{\dot{N}}}{\partial z^{B)}} = 0
\\
\label{cztery_kongruencje_row3}
&& \omega_{M\dot{N}} \, \frac{\partial f^{M}}{\partial z^{(\dot{A}}} \frac{\partial f^{\dot{N}}}{\partial z^{\dot{B})}} = 0
\end{eqnarray}
\end{subequations}
From (\ref{cztery_kongruencje_row2}) we obtain
\begin{equation}
\label{cztery_kongr_przejsciowe}
\omega_{M\dot{N}} \, \frac{\partial f^{M}}{\partial z^{A}} = \omega \, \frac{\partial f_{\dot{N}}}{\partial z^{A}}
\end{equation}
where $\omega = \omega (z^{A}, z^{\dot{B}})$ is some function. Contracting (\ref{cztery_kongr_przejsciowe}) with $\frac{\partial f_{N}}{\partial z_{A}}$ one gets
\begin{equation}
\frac{1}{2} \alpha_{1} \, \omega_{M \dot{N}} = \omega \, \frac{\partial f_{\dot{N}}}{\partial z^{A}}  \frac{\partial f_{M}}{\partial z_{A}} 
\end{equation}
Remembering, that $\alpha_{1} \ne 0$ and re-defining function $\omega$ we find the solution for $\omega_{M \dot{N}}$
\begin{equation}
\label{cztery_kongruencje_pomoc2}
 \omega_{M \dot{N}} =- 2\Omega \,  \frac{\partial f_{M}}{\partial z^{A}}   \frac{\partial f_{\dot{N}}}{\partial z_{A}} 
\end{equation}
Putting (\ref{cztery_kongruencje_pomoc2}) into (\ref{cztery_kongruencje_row1}) one finds the solution for $f_{A\dot{B}}$ and finally arrives at the metric (\ref{metryka_przestrzeni_z_czterema_kongruencjami_strun}). The last step is to feed (\ref{cztery_kongruencje_row3}) by (\ref{cztery_kongruencje_pomoc2}). The result is the system (\ref{rownania_jakie_zostaly_dla_metryki_cztery_kongr}).
\end{proof}

[\textbf{Remark.} One can first solve the Eq. (\ref{cztery_kongruencje_row3}). In this case we arrive at the equivalent form of the metric
\begin{equation}
\label{metryka_przestrzeni_z_czterema_kongruencjami_strun_equivalent}
ds^{2} = 2\Omega_{1} \, \left( \alpha_{4} \, \frac{\partial f^{M}}{\partial z^{A}} \frac{\partial f_{M}}{\partial z^{\dot{B}}} - \alpha_{2} \, \frac{\partial f^{\dot{M}}}{\partial z^{A}} \frac{\partial f_{\dot{M}}}{\partial z^{\dot{B}}} \right) dz^{A} dz^{\dot{B}}
\end{equation}
with the equations
\begin{equation}
\label{rownania_jakie_zostaly_dla_metryki_cztery_kongr_equivalent}
\frac{\partial f_{M}}{\partial z^{\dot{A}}} \frac{\partial f_{\dot{N}}}{\partial z_{\dot{A}}} 
\frac{\partial f^{M}}{\partial z^{(A}} \frac{\partial f^{\dot{N}}}{\partial z^{B)}} = 0 \ \ \ ]
\end{equation}

\section{Concluding remarks}

\begin{comment}
In this paper the relation between properties of cns and Petrov-Penrose types of SD Weyl spinor and algebraic types of the traceless Ricci tensor have been studied. It has been proved that the existence of at least two distinct cns has a significant effect on SD Weyl spinor and the traceless Ricci tensor. This effect is especially visible if at least one cns is nonexpanding. 
\end{comment}

There are several directions of further investigations. In this paper we considered only SD congruences. The manifold $\mathcal{M}$ can be equipped with similar structures on ASD side. Intersections of these congruences constitute the sngc. Because in $\mathbf{CR}$ and $\mathbf{UR}$, SD and ASD structures are unrelated one can consider the spaces with mixed Petrov-Penrose types (like $[\textrm{D}] \otimes [\textrm{N}]$ or $[\textrm{II}_{rc}] \otimes [\textrm{III}_{r}]$) equipped with many sngc with different properties of optical scalars (twist, expansion). Note, that the types $[\textrm{I}] \otimes [\textrm{I}]$ (in $\mathbf{CR}$) and $[\textrm{I}_{r}] \otimes [\textrm{I}_{r}]$ (in $\mathbf{UR}$) admits 4 SD and 4 ASD congruences of null strings. Intersections of these congruences give 16 distinct sngc. Do the spaces having so rich structure exist? Or maybe some (unrecognized yet) conditions make such spaces impossible to exist?

There are 6 different types of SD Weyl spinor in $\mathbf{CR}$ and 10 in $\mathbf{UR}$ (see \cite{Rod_Hill_Nurowski, Chudecki_Classification}). The existence of cns allows to distinguish the special subtypes. For example, type $[\textrm{I}]$ in general do not admit any cns. More special type $[\textrm{I}]^{e}$ admits one cns (superscript $e$ means that the corresponding cns is expanding while $n$ denotes nonexpanding cns). Types $[\textrm{I}]^{ee}$, $[\textrm{I}]^{eee}$ and $[\textrm{I}]^{eeee}$ are equipped with two, three and four cns, correspondingly. All possible subtypes of $C_{ABCD}$ in $\mathbf{CR}$ and $\mathbf{UR}$ are listed on the Schemes \ref{Degeneration_scheme_of_complex_case} and \ref{Degeneration_scheme_of_real_case}. ($[-]^{e}$ means, that all cns are expanding, while $[-]^{n}$ means, that at least one of them in nonexpanding).

\begin{Scheme}[!ht]
\begin{displaymath}
%\resizebox{1\textwidth}{!}{
\xymatrixcolsep{0.0cm}
\xymatrixrowsep{0.8cm}
\xymatrix{
\textrm{0 cns:} &  [\textrm{I}]  \ar[d] & [\textrm{II}]  \ar[d]  \ar[dr]  & & [\textrm{III}]  \ar[d] \ar[dr] &    &
[\textrm{D}]  \ar[d]  \ar[dr]  &  &  &  [\textrm{N}]  \ar[d]  \ar[dr] & & [-] \ar[ddddd] \ar[dddddr] \\ 
\textrm{1 cns:}  & [\textrm{I}]^{e} \ar[d]  &  [\textrm{II}]^{e} \ar[d] \ar[dr] & [\textrm{II}]^{n} \ar[d]  &  [\textrm{III}]^{e} \ar[d]  \ar[dr] & [\textrm{III}]^{n} \ar[d] & 
[\textrm{D}]^{e} \ar[d]  \ar[dr] & [\textrm{D}]^{n} \ar[d]  \ar[dr] &  &  [\textrm{N}]^{e}  & [\textrm{N}]^{n}  &  \\
\textrm{2 cns:}  & [\textrm{I}]^{ee}  \ar[d]  &  [\textrm{II}]^{ee} \ar[d] \ar[dr] & [\textrm{II}]^{en} \ar[d] &  [\textrm{III}]^{ee}  & [\textrm{III}]^{en}  & 
[\textrm{D}]^{ee}  & [\textrm{D}]^{en}  & [\textrm{D}]^{nn} &    &   &  \\
\textrm{3 cns:}  & [\textrm{I}]^{eee} \ar[d]  &  [\textrm{II}]^{eee}  & [\textrm{II}]^{een}  &  &  & &   &  &  &   &  \\
\textrm{4 cns:}  & [\textrm{I}]^{eeee}   &    & &  &  & &   &  &  &   &  \\
\textrm{$\infty$ cns:}  &    &    & &  &  & &   &  &  &   &  [-]^{e}  & [-]^{n} \\
}
%}
\end{displaymath} 
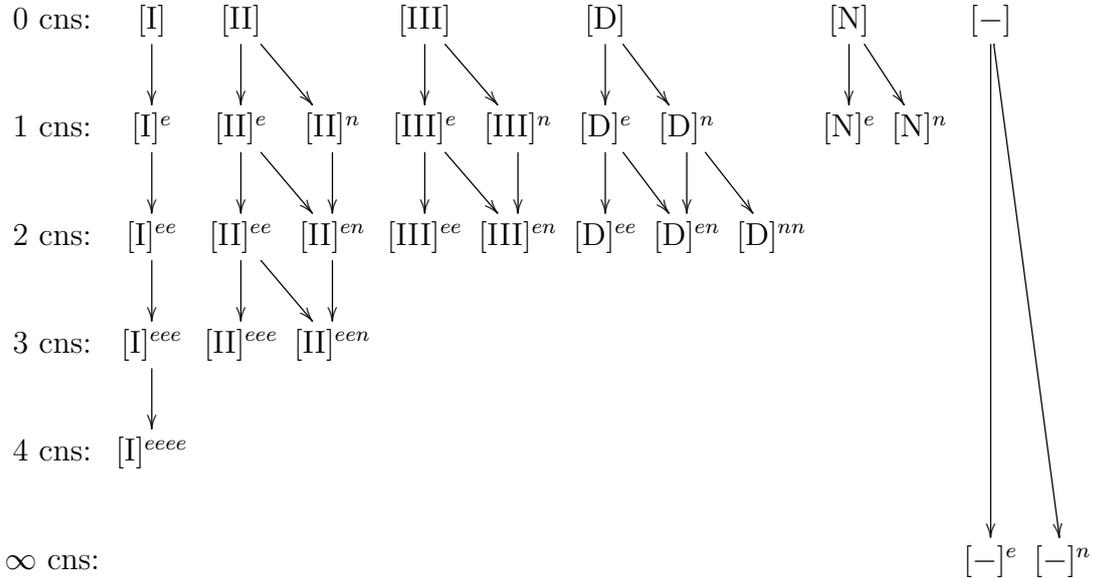
\caption{Types of $C_{ABCD}$ in $\mathbf{CR}$ equipped with different numbers of congruences of the SD null strings.}
\label{Degeneration_scheme_of_complex_case}
\end{Scheme}

\begin{Scheme}[!ht]
\begin{displaymath}
\resizebox{1\textwidth}{!}{
\xymatrixcolsep{0.0cm}
\xymatrixrowsep{0.8cm}
\xymatrix{
\textrm{0 cns:} &  [\textrm{I}_{r}]  \ar[d] & [\textrm{I}_{rc}]  \ar[d] & [\textrm{I}_{c}] & [\textrm{II}_{r}]  \ar[d]  \ar[dr]  & & [\textrm{II}_{rc}] \ar[d]  \ar[dr] & & [\textrm{III}_{r}]  \ar[d] \ar[dr] &    &
[\textrm{D}_{r}]  \ar[d]  \ar[dr]  &  &  & [\textrm{D}_{c}] & [\textrm{N}_{r}]  \ar[d]  \ar[dr] & & [-] \ar[ddddd] \ar[dddddr] \\ 
\textrm{1 cns:}  & [\textrm{I}_{r}]^{e} \ar[d]  & [\textrm{I}_{rc}]^{e} \ar[d] &  & [\textrm{II}_{r}]^{e} \ar[d] \ar[dr] & [\textrm{II}_{r}]^{n} \ar[d]  & [\textrm{II}_{rc}]^{e} & [\textrm{II}_{rc}]^{n} & [\textrm{III}_{r}]^{e} \ar[d]  \ar[dr] & [\textrm{III}_{r}]^{n} \ar[d] & 
[\textrm{D}_{r}]^{e} \ar[d]  \ar[dr] & [\textrm{D}_{r}]^{n} \ar[d]  \ar[dr] &  & & [\textrm{N}_{r}]^{e}  & [\textrm{N}_{r}]^{n}  &  \\
\textrm{2 cns:}  & [\textrm{I}_{r}]^{ee}  \ar[d] & [\textrm{I}_{rc}]^{ee}  & &  [\textrm{II}_{r}]^{ee} \ar[d] \ar[dr] & [\textrm{II}_{r}]^{en} \ar[d] & & &  [\textrm{III}_{r}]^{ee}  & [\textrm{III}_{r}]^{en}  & 
[\textrm{D}_{r}]^{ee}  & [\textrm{D}_{r}]^{en}  & [\textrm{D}_{r}]^{nn} &  &  &   &  \\
\textrm{3 cns:}  & [\textrm{I}_{r}]^{eee} \ar[d]  & & &  [\textrm{II}_{r}]^{eee}  & [\textrm{II}_{r}]^{een} & & &  &  & &   &  &  &  & &  \\
\textrm{4 cns:}  & [\textrm{I}_{r}]^{eeee} & &  & & &  & &  &  & &  & &  &  &   &  \\
\textrm{$\infty$ cns:}  &    & & &   & &  &  & &  & & &  &  &  & &  [-]^{e}  & [-]^{n} \\
}
}
\end{displaymath} 
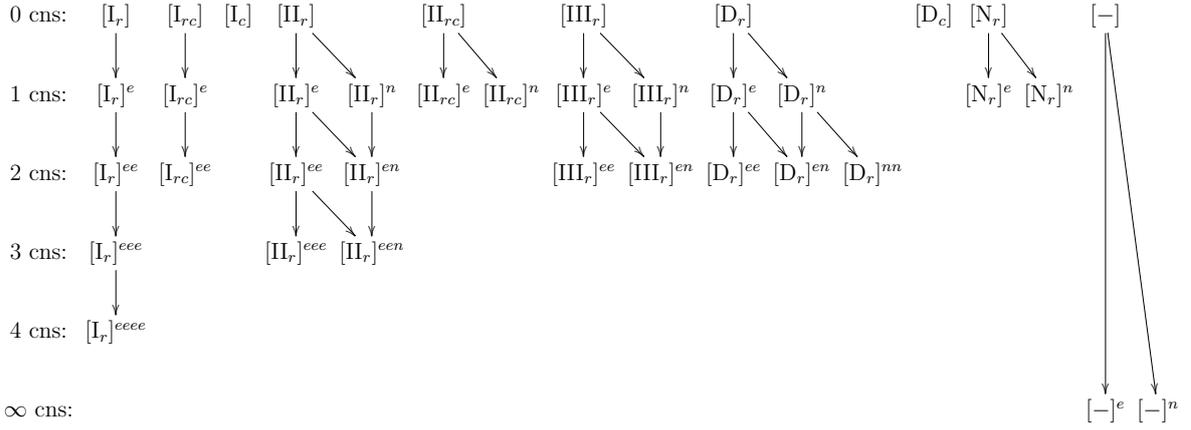
\caption{Types of $C_{ABCD}$ in $\mathbf{UR}$ equipped with different numbers of congruences of the SD null strings.}
\label{Degeneration_scheme_of_real_case}
\end{Scheme}

The possible applications of our results concern not only $\mathbf{CR}$ and $\mathbf{UR}$ but also $\mathbf{HR}$. In \cite{Nurowski_Trautman} there is an example of Lorentzian manifold equipped with three distinct sngc (which are the intersection of 3 SD and 3 ASD congruences of null strings). The general form of the metric of the space equipped with three distinct sngc is still unknown. The richest possible structure appears as intersections of 4 SD and 4 ASD congruences of null strings. In $\mathbf{HR}$ it gives 4 distinct sngc. In \cite{Trautman} Trautman asks: "In dimension 4, are there not conformally flat Lorentz manifolds that have 4 distinct shear-free null geodesic congruences?". Such spaces (if exist) cannot be conformal to Einstein spaces (because of the Goldberg - Sachs theorem). Our considerations (Theorem \ref{twierdzenie_o_spinorze_generujacym_wstege_i_spinorze_Penrosa}) give more properties of such spaces: they must be of Petrov-Penrose type [I]. The explicit metric as well as the possible types of tensor of matter admitted by such spaces are unknown. The problem of Lorentzian manifolds with 4 distinct sngc is very interesting and we are going to analyze it in details in near future. 

Quite surprisingly, the space with two distinct expanding cns seems to admit all algebraic types of the traceless Ricci tensor. We believe, that in fact there are some restrictions which reduce the number of possible types in such a case. Similarly, in spaces equipped with three or four expanding cns those restrictions should be even stronger. Frankly, this problem is more advanced and it involves deeper analysis.

Finally, expansion of the cns is the most transparent property of such structure and has clear geometrical interpretation. It seems, that the Sommers vector carries some information about cns. Natural question arises if the properties the Sommers vector define a new subclassification of cns? 
\newline
\newline
\textbf{Acknowledgments}
\newline
Some points of the present paper were presented in July 2016 in Brno at the 13th conference \textsl{Differential Geometry and its Applications} and in September 2016 in Krakow at \textsl{The 3rd Conference of the Polish Society on Relativity}. The author is indebted to Maciej Przanowski for his interest in this work and for help in many crucial matters.

%#####################################################################################

\end{document}